\documentclass[preprint, 3p]{elsarticle}

\usepackage{amssymb}
\usepackage{amsmath}
\usepackage{amsthm}

\usepackage{thmtools, thm-restate}
\usepackage{xcolor}
\usepackage[colorlinks=true]{hyperref}

\usepackage[nameinlink,noabbrev]{cleveref}

\newtheorem{theorem}{Theorem}[section]
\newtheorem{lemma}[theorem]{Lemma}
\newtheorem{definition}[theorem]{Definition}

\newtheorem{corollary}[theorem]{Corollary}

\newcommand{\E}{\mathbb{E}}

\journal{IFIP WG 7.3 Performance 2026}

\begin{document}

\begin{frontmatter}

\title{Pareto-Optimal Scheduling in the Half-batch Multiserver-job Model}

\author[nu]{Ziyuan Wang}
\author[nu]{Izzy Grosof}
\affiliation[nu]{organization={Northwestern University},
            addressline={Deparment of Industrial Engineering and Management Science},
            city={Evanston},
            postcode={60201},
            state={IL},
            country={USA\\},
            email = {ziyuanwang2027@u.northwestern.edu, izzy.grosof@northwestern.edu}}

\begin{abstract}
In large-scale computing systems, jobs often demand heterogeneous server allocations: large jobs that occupy a substantial fraction of the servers are of high importance and are thus latency-sensitive, while small jobs fill in the remaining capacity to maintain throughput. To model this dynamic, we introduce the half-batch multiserver-job (MSJ) framework, a queueing model in which large jobs arrive according to a Poisson process and require all servers simultaneously, while small jobs, each needing only one server, are always available.

We prove that, in the half-batch MSJ model, the Pareto frontier for large-job mean response time and small-job throughput admits a simple and exact characterization. It is generated by a family of convoy policies, under which the system serves small jobs until $k$ large jobs have arrived and then switches to serving large jobs, together with convex combinations of neighboring convoy policies. Our result is fully general and non-asymptotic, holding for every stable arrival rate $\lambda$, every number of servers $n$, and every large-job size distribution $S$.
\end{abstract}
 


\end{frontmatter}

\section{Introduction}
\label{sec:intro}

In modern large-scale computing systems, jobs often have highly heterogeneous resource requirements. This is especially clear in supercomputing systems such as Frontier at Oak Ridge \cite{FrontierUserGuide}, where the system is built primarily to serve very large jobs, often called leadership-class jobs, that request more than half of the existing nodes. At the same time, smaller jobs are used to fill otherwise idle capacity.

This setting is naturally modeled by the multiserver-job (MSJ) framework \cite{HarcholBalter2022}. In these systems, a job requests a number of servers or nodes and occupies them simultaneously throughout service. This distinguishes the setting both from classical single-server models, which cannot represent simultaneous server needs, and from multiresource models, where multiple resource dimensions are needed. Similar phenomena also arise in datacenters and large-scale machine learning workloads, where a relatively small number of important, latency-sensitive large jobs are given priority while coexisting with a much larger population of low-resource jobs.

The presence of these two job classes creates a trade-off between small-job throughput and large-job latency. Moreover, conventional queueing models do not capture this trade-off well. In a standard open system, a large throughput region is typically a prerequisite for good latency, while in a closed system latency and throughput are linked by Little's law. Here, by contrast, the central question is how to design a policy that balances latency against throughput.

A key ingredient of this trade-off is non-preemptive scheduling. If preemption were free, then one could simply run small jobs whenever capacity is available and interrupt them immediately when a large job needs to start. In practice, however, preemption is often expensive or unavailable. A running small job may have to be killed, checkpointed, or allowed to complete at the expense of large jobs.

Existing studies of nonpreemptive policies do not fully capture this trade-off. Some nonpreemptive throughput-oriented policies, such as Randomized Timers \cite{Ghaderi2016} and the Markovian Service Rate policy \cite{Chen2025}, can achieve strong stability guarantees by building very long convoys of large jobs. These policies accumulate many jobs of identical resource requirements, and then serve all of them in sequence in a very long convoy. Such policies may be unacceptable in practice because of the latency they impose.

To study systems in which this trade-off is important, we introduce a novel MSJ model, which we call the half-batch MSJ model, in which large jobs arrive according to a Poisson process while small jobs are always available. The term ``batch'' comes from batch processing, meaning that the small jobs represent background work that is run when resources are available. Within this model, many policies are possible, including randomized policies such as MSR, as well as policies that switch to large jobs only after a predetermined number of small jobs have completed. This leads to the central question of the paper:

\begin{center}
\begin{quote}
\textit{Which family of policies form the Pareto-optimal trade-off between small-job throughput and large-job latency?}
\end{quote}
\end{center}

We prove that the Pareto frontier admits a simple and exact characterization. It is generated by a family of convoy policies: under a $(k,n)$-convoy policy, the system runs small jobs on all $n$ servers until the $k$th large job arrives, then clears the small jobs as quickly as possible and serves large jobs exhaustively until the large-job queue becomes empty. We show that every point on the Pareto-optimal curve is achieved either by one of these $(k,n)$-convoy policies or by randomizing between two neighboring convoy policies, together with the endpoint that serves no small jobs at all and begins large-job service immediately.

Our result is of interest because this kind of complete characterization has previously appeared mainly in single-server settings, for example under the Shortest Remaining Processing Time (SRPT) policy, and, has not been obtained in an MSJ setting before. In addition, our optimality result characterizes the Pareto frontier of a bicriterion objective, whereas most previous optimality results concern a single performance metric. The result is fully general and non-asymptotic, holding for every stable arrival rate $\lambda$, every number of servers $n$, and every large-job size distribution $S$.

Our proof techniques are also of independent interest. In many ways, the proof can be viewed as a Markovian analogue of an interchange argument: we construct a sequence of improvements that inexorably pushes any policy toward a very precise family of policies. A key step is proving that the relevant performance metrics are concave in the number of small jobs in service, which forces the policy either to continue building the convoy or to serve the convoy by bringing large jobs into service as quickly as possible.

\subsection{Outline}

The rest of the paper is organized as follows. \Cref{sec:prior} reviews related work. \Cref{sec:model} introduces the half-batch MSJ model and the $(k,m)$-convoy policy. \Cref{sec:main} presents our main results and provides a proof outline. \Cref{sec:numerical} numerically examines the robustness of our policies beyond the assumptions of the theoretical model. \Cref{sec:proof} proves the main theorem and the supporting lemmas needed for the proof, organized into three subsections. 

\section{Prior Work}
\label{sec:prior}

Our half-batch MSJ model is a novel hybrid of open and closed systems, combining external Poisson arrivals with an internal, always-available stream of small jobs. While both open (\Cref{sec:open}) and closed (\Cref{sec:closed}) MSJ models have been extensively studied, the half-batch model naturally bridges these two settings.

\subsection{Open models}
\label{sec:open}
There have been several recent results on open MSJ models, i.e., models with only external Poisson arrivals. A number of throughput-optimal policies have been studied, including preemptive policies such as MaxWeight \cite{Maguluri2012} and ServerFilling \cite{Grosof2022WCFS}, as well as non-preemptive policies such as Randomized Timers \cite{Ghaderi2016, Konstantinos2018}. The recently introduced Markovian Service Rate policy class includes both preemptive and non-preemptive variants and is likewise throughput optimal \cite{Chen2025}.

For response time analysis, the lack of work conservation in the multiserver-job model makes it challenging to analyze even simple policies such as FCFS. Only recently has the first explicit characterization of the mean response time in the MSJ FCFS system been obtained \cite{Grosof2023TRaM}. That work also uses closed models in its analysis. Additionally, a variant of SRPT called ServerFilling-SRPT has been shown to be optimal in the heavy-traffic regime \cite{Grosof2022MSJ}. The MSR-type policies also admit mean response time formulas \cite{Chen2025}. Note that throughput optimality is a prerequisite for optimizing response time.

The MSJ FCFS model has also been studied in asymptotic regimes where the number of servers, the arrival rate, and the server-need distribution all grow to infinity. When these system parameters scale together, \citet{Hong2023} establish bounds on the mean waiting time in the system. On the other hand, \citet{Grosof2026Mrtu} consider a load-focused multilevel scaling regime, where the heavy-traffic limit is taken before the secondary scaling $n\to\infty$ and $p_n\to0$. The model studied by \citet{Grosof2026Mrtu} also has a two-class server-need structure similar to ours, in which jobs require either one server or all $n$ servers.

In the open MSJ model, stability maximization becomes a prerequisite for optimizing response time in the heavy-traffic regime. In contrast, the half-batch model allows response time and throughput to be studied separately.

\subsection{Closed models}
\label{sec:closed}
Closed or saturated MSJ models are also of particular interest, as the throughput of the saturated system matches the stability threshold of the corresponding open system, a folk theorem first formally proven by \citet{Baccelli1995}. \citet{Rumyantsev2017} implicitly used this connection through a matrix-analytic approach, while \citet{Grosof2023Saturated} were the first to make it explicit, deriving a product form expression for the state probabilities of the MSJ system under FCFS. They consider both the case in which all job types have exponentially distributed service durations and a two-class setting where each class follows a different exponential distribution. Subsequent work has extended product form stationary distributions to graph-based product forms \cite{Comte2026}. In a saturated system, latency is not well defined, while in a closed system, latency and throughput are equivalent by Little's law.

\citet{Olliaro2023} study a saturated system with two job classes: small jobs that require one server and large jobs that require all servers. In their setting, small jobs follow an exponential distribution, while large jobs may follow an arbitrary distribution. Their model is closely related to ours. However, they focus on a fully saturated system with the goal of characterizing the stability region of the corresponding open system, whereas we study a hybrid system, half-open and half-closed, with the goal of deriving the optimal trade-off between small-job throughput and large-job latency.

\subsection{Semi-open models}

A novel aspect of our MSJ model is its arrival structure, which includes both an external arrival stream and internal batched jobs. While this specific setting has not been previously studied in the queueing theory literature, systems with both external arrivals and internally circulating jobs have received significant attention due to their theoretical significance and practical relevance. Two prominent examples are Jackson networks \cite{Jackson1963,Baccelli1994} and semi-open queueing networks (SOQNs) \cite{Jia2009,Roy2016}.

In the half-batch MSJ model, small jobs are always available in a saturated manner (and are unconstrained), and the number of small jobs in service is controlled by the policy. In SOQNs, the internal component is stochastically constrained, whereas in Jackson networks all jobs, whether internal or external, follow a Markovian routing structure without a fixed population constraint.

\citet{Bhulai2003} study a call-blending model with arriving
inbound calls and always-available outbound calls, maximizing
outbound throughput subject to a constraint on mean inbound
waiting time. Unlike our MSJ model, each call requires only
one server.

\subsection{Bi-criterion optimization}  
Prior work in broader queueing and resource-allocation settings has also explored multi-objective trade-offs. For example, \citet{Philipp2022} analyze a multi-class, multi-server FCFS-ALIS matching queue, where FCFS-ALIS stands for ``First Come, First Served - Assign Longest Idle Server.'' They decompose the compatibility graph into resource-pooling components and study a Pareto frontier that balances average customer waiting time against aggregate matching reward. However, their focus is on designing matching topologies rather than developing control policies, which is central to our paper. 

\citet{Dumitru2014} model data-intensive bag-of-tasks execution on cloud VMs as a closed BCMP network with shared-bandwidth downloads. Their frontier balances makespan against monetary execution cost.

These trade-offs bear similarities to the latency-throughput trade-off studied in our work. In many server farms and cloud computing platforms that rent out their resources, throughput directly translates into reward or monetary value. However, to the best of our knowledge, no prior work has studied this trade-off in the MSJ setting.

\subsection{Single-server policies}
\label{sec:lit_4}
Part of our analysis involves studying specific variants of the $M/G/1$ queue. This includes the $M/G/1$ queue under the so-called $N$-policy, in which the server begins service only after $N$ jobs have accumulated in the system \cite{Yadin1963}. Since $n$ denotes the number of servers in this work, this policy will be referred to as the $k$-policy for to avoid confusion. 

Another queue of interest is the $M/G/1$ queue with setup cost, whose response time was first derived by \citet{Welch1964} and is also discussed by \citet{Harchol-Balter2013}. The setup cost refers to an initial delay experienced by the first job of each busy period before service can begin. These two models are combined to derive the response time of an $M/G/1$ queue with setup cost operating under the $k$-policy (\Cref{sec:kmconvoy}), which we then use to characterize the response time of the convoy policies studied in this paper.

\section{Model}
\label{sec:model}

\subsection{Half-batch model}

We consider a model with $n$ servers and two job classes. \emph{Large} jobs are time-sensitive and require all $n$ servers simultaneously, whereas \emph{small} jobs are the time-insensitive ``batch jobs'' and use a single server each. At any time, the system can serve either one large job or up to $n$ small jobs in parallel. Both classes are non-preemptible. Consequently, if any servers are occupied by small jobs, a large job can begin service only after all small jobs currently in service have completed.

Small jobs have an exponential service-time distribution $S_1\sim\exp(\mu_1)$ with mean $\E[S_1]=1/\mu_1$. The system is saturated with small jobs in the sense that additional small jobs are always available and may be put into service whenever there are idle servers. The exponential assumption is needed for tractability, as our analysis often uses the memoryless property of the exponential distribution to study the long-run distribution of the number of small jobs in the system under certain simple policies.

Large jobs arrive according to a Poisson process with rate $\lambda$ and are served in First-Come, First-Served (FCFS) order. In contrast to small jobs, large jobs may follow an arbitrary service-time distribution $S$ with mean $\E[S]=1$, which is normalized. The system load is $\rho=\lambda \E[S]$, corresponding to the load contributed by large jobs. Under the normalization $\E[S]=1$, this reduces to $\rho=\lambda$. Throughout, we assume $\rho<1$.

\subsection{Control policies}
\label{sec:schedule}

A policy $\pi$ specifies when to admit large or small jobs into service, subject to the constraint that jobs can be added only when enough servers are available. Additionally, we assume the policies are size-blind with respect to all the jobs. The system is said to be in level $m\in\{0,1,\ldots,n\}$ if it is currently serving $m$ small jobs.

We study the trade-off between two primary performance metrics: the mean response time of large jobs, denoted $\E[T^\pi]$, and the throughput of small jobs, denoted $X^\pi$. 

We say that a policy renews when it enters the state where all large jobs have completed and no small jobs have been admitted, even if it instantaneously passes through this state. We consider only policies that are positive recurrent with respect to this state, if a policy is not positive recurrent, then a positive fraction of time is spent with at least one extra large job contributing cumulative response time, which makes the policy suboptimal. Additionally, positive recurrence implies that the renewal-based performance metrics we study are well defined.

 Furthermore, positive recurrence to any state other than the empty state is also suboptimal, since any additional large jobs present contribute extra cumulative response time, while any admitted small jobs are controlled by the policy.

For a given policy $\pi$, define a renewal period as the interval between two successive visits to the renewal state. Equivalently, the process renews when the system becomes free of large jobs, immediately before any small jobs are placed into service.

We also work with auxiliary quantities, including the mean number of small jobs completed in a cycle, denoted $\E[N^\pi]$, the mean cycle length $\E[L^\pi]$, and the mean total cumulative response time in a cycle, $\E[C^\pi]
=
\sum_{j\in\mathcal{B}}\E[T_j^\pi],$
where $\mathcal{B}$ denotes the set of large jobs that arrive during a renewal cycle. A useful identity is $\E[C^\pi]=\lambda\,\E[L^\pi]\cdot \E[T^\pi],$
which follows from Little's law $\E[N_L^\pi]=\lambda\E[T^\pi]$ together with the renewal-reward relation $\E[N_L^\pi]=\E[C^\pi]/\E[L^\pi]$.

Finally, a policy $\pi^\star$ is Pareto optimal for the objective pair $(\E[T],X)$ if there is no policy $\pi'$ such that $\E[T^{\pi'}]\leq\E[T^{\pi^\star}]$ and $X^{\pi'}\geq X^{\pi^\star}$ with at least one inequality strict.

\subsection{Replication}
\label{sec:replication}

One of the key steps in establishing Pareto optimality is to show that any Pareto-optimal policy must belong to a class of policies constructed from basic building blocks, which we call \textit{threshold policies}. These are short-term policies that describe the behavior of a larger policy between two consecutive large-job arrivals, applied when there is no large job in service.

Specifically, we consider two classes of threshold policies, called \emph{Hold-at-$\ell$} and \emph{Drop-to-$\ell$}, indexed by a level $\ell\in\{0,1,\ldots,n\}$.

The Hold-at-$\ell$ policy is a short-term policy that can be applied when the number of small jobs in service at the first large-job arrival is less than $\ell$. Under this policy, the scheduler immediately increases the number of small jobs in service to $\ell$ and then holds the system at that level until the next large-job arrival.

The Drop-to-$\ell$ policy can be applied when, at the first large-job arrival, there are at least $\ell$ small jobs in service. Under this policy, the scheduler reduces the number of small jobs in service to level $\ell$ as quickly as possible through successive small-job completions, without admitting any additional small jobs into service, and then holds the system at that level until the next large-job arrival.

When constructing from these building blocks, a Hold-at-$\ell$ or Drop-to-$\ell$ policy is selected at each large-job arrival according to a probabilistic mixture. We show that any policy can be replicated by such a mixture of threshold policies, where replication means that all key performance metrics are matched exactly, see \Cref{sec:replication_proofs}.

\subsection{$(k,m)$-convoy}
\label{sec:kmconvoy}

Our main result gives an exact characterization of the Pareto frontier. In particular, we show that the frontier consists only of a certain restricted class of policies, which we now define.

For a positive integer $k$ and $m\in\{0,1,\ldots,n\}$, we define the $(k,m)$-convoy policy as follows: the system runs small jobs on $m$ servers until the arrival of the $k$-th large job. At that point, the system enters a \textit{clearing phase} in which all small jobs are completed as quickly as possible to make room for large jobs, after which the large jobs are served exhaustively as a convoy, one after another, until no large jobs remain. The policy then renews.

The duration of this clearing phase corresponds to the maximum of $m$ independent exponential service times, which we denote by $I_m$.

For the special case $m=n$, in which all servers initially run small jobs, we prove in \Cref{thm:main} that the resulting $(k,n)$-convoy policies generate the Pareto frontier.

The $(k,m)$-convoy policy also connects naturally to several classical $M/G/1$ models from the perspective of a large job. When $k=1$ and $m>0$, the small jobs act as a setup cost for the large job \cite{Harchol-Balter2013}. When $m=0$ and $k>1$, the policy reduces to the classical $N$-policy (\cite{Yadin1963}), with $k=N$, see \Cref{sec:lit_4}. Finally, when $k=1$ and $m=0$, we recover an $M/G/1$ queue from the perspective of large jobs. Since the only large-job performance metrics we consider are $\E[T]$ and $\E[C]$, a classical result implies that, among non-preemptive and size-blind policies, the service order is irrelevant \cite[Theorem~29.2]{Harchol-Balter2013}. Thus, in this case, we recover the standard $M/G/1$ FCFS queue with service-time distribution equal to the large-job size distribution, which we denote by $M/G/1$-$L$.

\section{Main Results}
\label{sec:main}
We now state the main result of the paper, which gives a complete characterization of the Pareto frontier in the half-batch MSJ model. The empirical valuation can be found in \Cref{sec:numerical} and the proofs can all be found in \Cref{sec:proof}.

\begin{theorem}
\label{thm:main}
In the half-batch MSJ system, for any arrival rate of large jobs $\lambda$ and any large-job size distribution $S$, the $(1,0)$-convoy policy and the $(k,n)$-convoy policies for $k \ge 1$ lie on the Pareto frontier of $(\E[T],X)$.

Furthermore, every point on the Pareto frontier can be achieved by a policy which is a convex combination of $(k,n)$ and $(k+1,n)$-convoy policies, or as a convex combination of the $(1,0)$ and $(1,n)$-convoy policies.
\end{theorem}

\begin{proof}[Proof deferred to \Cref{sec:proof_main}]
    
\end{proof}

\subsection{Proof Overview}
\label{sec:overview}

To prove \Cref{thm:main}, we start with an arbitrary policy $\pi$ and successively shrink the class of candidate policies toward the target family, namely, convex combinations of neighboring $(k,n)$-convoy policies and the convex combinations of the $(1,0)$- and $(1,n)$-convoy policies, through a sequence of optimality results.

\textit{Renewal cycles.} The first step concerns the renewal-cycle structure of a policy $\pi$. We define a renewal cycle as the period between successive visits to the state in which no large jobs are present and no small jobs have yet been admitted into service. In \Cref{sec:schedule}, we argued that an optimal policy must be positive recurrent with respect to this state, and that recurrence to any other state is suboptimal. In addition, \Cref{lem:exhaustion} shows that once an optimal policy begins serving large jobs, it must continue until all large jobs in the system, including those that arrive during large-job service, have completed before admitting small jobs again.

\textit{Interarrival replications.} Next, we show that for an arbitrary policy $\pi$, over any single interarrival period between two large-job arrivals, we can construct a probabilistic mixture of threshold policies that replicates both the behavior of $\pi$ during that period and the arrival-seen distribution at the next large-job arrival; see \Cref{thm:2_1}. See \Cref{sec:replication} for the definitions of threshold policies and replication.

\Cref{thm:2_1} gives the general one-interarrival replication result. Applying it successively over interarrival periods yields a full replication of an arbitrary policy by threshold policies. We call the resulting policy the Mixture-of-Threshold-$\pi$ policy, abbreviated MoT-$\pi$. By construction, MoT-$\pi$ is a probabilistic policy, choosing a threshold policy at each large-job arrival according to an appropriate distribution. This representation applies to any policy $\pi$, even if $\pi$ itself is deterministic, and is formalized in \Cref{cor:2_1}.

\textit{Suboptimality before clearing.} At this point, the original policy $\pi$ has been reduced to an equivalent MoT-$\pi$ policy. In this representation, the possible combinations of threshold-policy choices over the penultimate interarrival period and the subsequent clearing phase, fall into a small number of canonical forms; see \Cref{sec:suboptimality_proof}.

In particular, we show that, apart from the few penultimate decisions that lead to the $(k,n)$-convoy policies, any other penultimate choice can be improved by a convex combination of two policies that otherwise agree with MoT-$\pi$ and differ only over the penultimate interarrival period and the subsequent clearing phase. These two policies correspond to the neighboring $(k,n)$-convoy policies. The mixing coefficient $\alpha$ can be constructed explicitly, as given in \Cref{lem:3_1,lem:3_4,lem:3_7}, and the fact that the resulting convex combination Pareto dominates MoT-$\pi$, and hence also $\pi$, follows from \Cref{lem:3_6,lem:3_2,lem:3_8,lem:3_9}. This improved policy is called MoT-Mix-$\pi$ and is defined in \Cref{def:motmix}.

Moreover, once a MoT-$\pi$ policy reaches level $0$, it should begin serving large jobs immediately. Any further delay would worsen both small-job throughput and large-job response time.

Note that this involves two successive decisions by the policy at neighboring large-job arrivals. However, thanks to the construction of MoT-$\pi$ policies, the probability of each such joint branch can be computed explicitly, and on each possible branch the policy performs the corresponding MoT-Mix-$\pi$ modification.

Once a threshold policy decreases the small-job level, the threshold policy selected at the next large-job arrival should not increase it again. We call such behavior reneging, and it is suboptimal. The reason is that the probability of future reneging can be anticipated by the construction of MoT-$\pi$ policies and absorbed into the earlier decision, so that one can avoid reducing the small-job level in the first place, improving throughput without increasing large-job response time. The proof of which can be found in \Cref{sec:proof_main}.

Thus, by the no-reneging property, MoT-Mix-$\pi$ cannot involve an increase at a large arrival following a decrease in an earlier interarrival period. Therefore, an optimal policy must either select Hold-at-$n$ in the interarrival period before the clearing phase, or select Drop-to-$0$ and enter the clearing phase immediately. The no-reneging property then implies that before this decision the policy must always Hold-at-$n$. Once the level reaches $0$, the policy must immediately serve all large jobs before admitting small jobs again. Hence, we conclude that any Pareto-optimal policy must be either a $(k,n)$-convoy policy or a convex combination of many such policies. In addition, \Cref{lem:pareto_hypograph} shows that a convex combination of non-adjacent policies lies strictly below the Pareto curve and is therefore suboptimal, leaving only $(1,0)$- and $(k,n)$-convoy policies or a convex combination of adjacent convoy policies, completing the proof of \Cref{thm:main}.

\section{Numerical Robustness Analysis}
\label{sec:numerical}

Our theoretical results rely on several structural assumptions that make the model analytically tractable. In this section, we use simulation to examine the robustness of the proposed policies when these assumptions are relaxed. In particular, we consider nonexponential small-job service times and settings in which large jobs do not necessarily occupy all servers.

\subsection{Simulation setup and benchmark policies}

We first introduce the simulation setup and the policies that will be used throughout the numerical robustness analysis. To study the effect of relaxing the exponential small-job assumption, we model small-job service times using the Weibull family. Specifically, let the small-job service time $S_1$ have survival function $\mathbb{P}(S_1>x)=\exp\left\{-\left(\frac{x}{\theta}\right)^\alpha\right\},$ where $\alpha>0$ is the shape parameter and $\theta>0$ is chosen so that $\mathbb{E}[S_1]=1/\mu_1$. When $\alpha=1$, the Weibull distribution reduces to the exponential distribution with rate $\mu_1$, corresponding exactly to the setting covered by our theoretical results.

We begin with the benchmark case $\alpha=1$. This allows us to compare the simulated policies against the analytically derived Pareto-optimal frontier before varying $\alpha$ in the next subsection to examine robustness to nonexponential small-job service times.

In \Cref{fig:weibull100}, the blue dots represent the $(1,0)$-convoy policy (leftmost point), and the $(k,n)$-convoy policies. These points are connected by straight line segments corresponding to convex combinations of adjacent policies, forming the Pareto frontier. By our main result, \Cref{thm:main}, every point on this frontier is achieved by a convex combination of neighboring $(k,n)$-convoy policies.

In this example, $k$ ranges from $1$ to $5$, and we set $n=10$, $\mu_1=2$, and $\lambda=0.5$, although our results hold for all parameter values. The large-job sizes are exponentially distributed, even though this assumption is not required by our result, which allows an arbitrary large-job size distribution. These frontier points are not obtained by simulation and they are computed directly from the response-time and throughput formulas in \Cref{prop:kmconvoy}. As expected, the resulting curve is concave, and the region below it is convex. The curve approaches the limiting throughput $\lim_{k\to\infty}X^{k,n}=n\mu_1(1-\lambda)=10$.

\begin{figure}[h]
    \centering
    \includegraphics[width=0.8\textwidth]{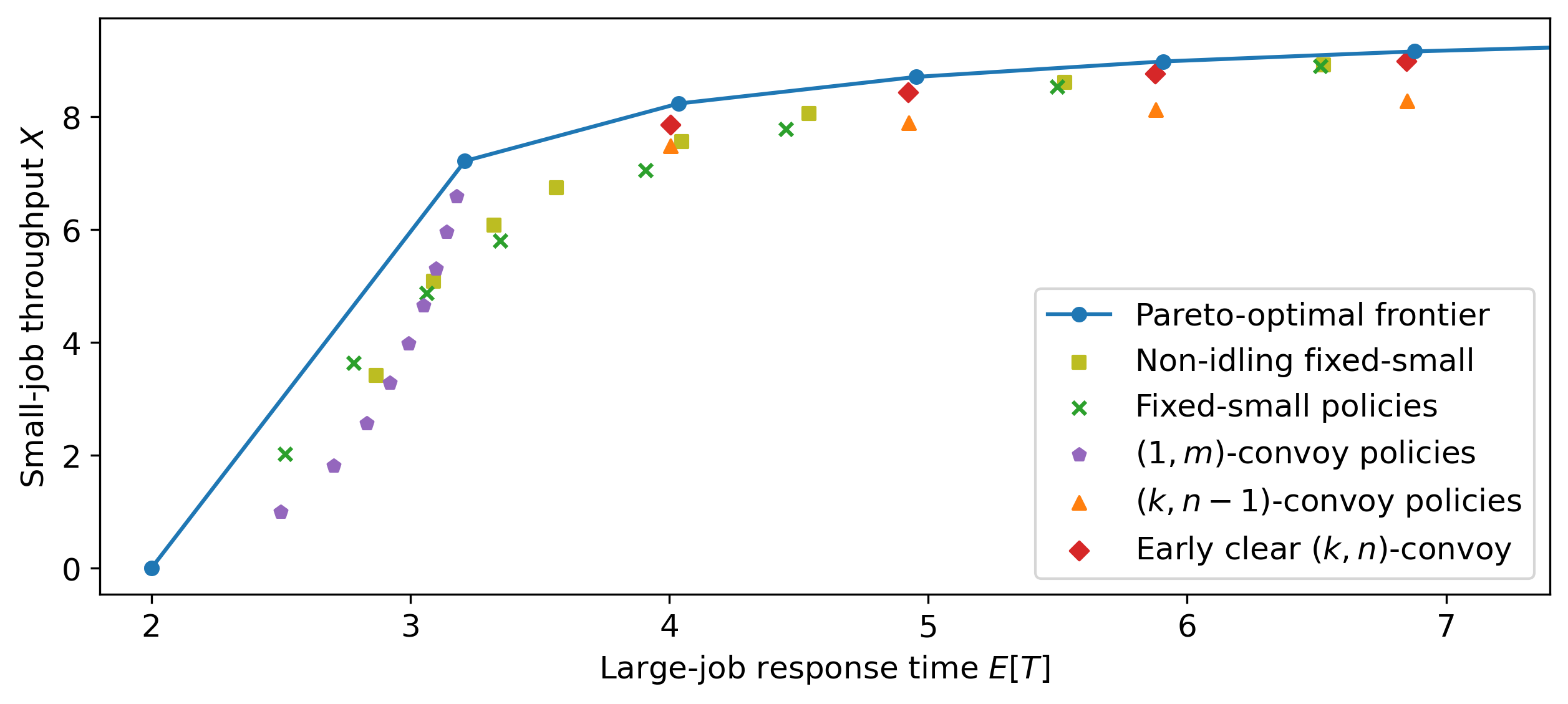}
    \caption{Comparison of the Pareto frontier with other suboptimal policies. The fixed-small, non-idling fixed-small, and $(k,n)$-convoy policies with early clearing are simulated over $5\times 10^6$ renewal cycles. The parameters are $n=10$, $\mu_1=2$, and $\lambda=0.5$. Small-job service times are Weibull with shape parameter $\alpha=1$, corresponding to exponential service times.}
    \label{fig:weibull100}
\end{figure}

Next, we plot the performance of several suboptimal $(k,m)$-convoy policies. These include the $(1,m)$-convoy policies with $m\in\{1,\ldots,n-1\}$, represented by the purple pentagons. These points form a convex curve connecting the $(1,0)$- and $(1,n)$-convoy policies. We also plot the $(k,n-1)$-convoy policies for $k\ge 2$, represented by the orange triangles. Both families lie strictly below the Pareto frontier.

We also plot three simulated policies. The first is the fixed-small policy. Under this policy, the system completes a fixed number $n_1$ of small jobs before switching to exhaustive service of all waiting large jobs. If no large job is present after the $n_1$ small-job completions, the system idles until the next large-job arrival and then serves large jobs exhaustively. Once the large-job queue is empty, the system admits and serves a new batch of $n_1$ small jobs.

The simulated performance of the fixed-small policies is shown by the green crosses, with $n_1$ ranging from $10$ to $160$. Each point is simulated using $5\times 10^6$ renewal cycles. As $n_1\to 0$, the policy converges to the $(1,0)$-convoy policy. As $n_1$ grows, the policy appears to move closer to the Pareto frontier and may be asymptotically optimal. This is intuitive, because as $n_1$ grows the probability of idling goes to zero, and the policy behaves like a convex combination of many $(k,n)$-convoy policies with large $k$, hovering near the limiting throughput value of $10$.

The second simulated policy is a non-idling variant of the fixed-small policy. Under this policy, after the admitted $n_1$ small jobs have completed, if no large jobs are present, the system admits and serves a new batch of $n_1$ small jobs instead of idling until the next large-job arrival. These policies are represented by the squares in \Cref{fig:weibull100}, using the same values of $n_1$ as the fixed-small policy, ranging from $10$ to $160$, and are simulated with $10^8$ large-job arrivals. As $n_1$ increases, the squares coincide with the crosses representing the idling fixed-small policies, as expected, since the probability of idling vanishes as $n_1$ grows.

The third policy is an improvement over the suboptimal $(k,n-1)$-convoy policies (orange triangles). This policy attempts to anticipate the arrival of the $k$th large job by clearing small jobs early. We call any policies that does this \textit{anticipatory policies}. In this example, we simulate the case where the policy reduces the small-job level from $n$ to $n-1$ if the first small-job completion between the $(k-1)$st and $k$th large-job arrivals occurs before the next large-job arrival. These policies are simulated with $10^8$ large-job arrivals and are shown by the red diamonds.

Again, as $k$ grows, the policy appears to become asymptotically optimal, since the loss from idling one server during a single interarrival period is washed out over many interarrival periods. One can similarly clear more small jobs early, which would trace out a curve from the $(k,n)$-convoy policy toward the $(k-1,n)$-convoy policy.

\subsection{Nonexponential small jobs}

We next examine the robustness of the convoy policies when the small-job
service times are nonexponential. We consider Weibull service times with
shape parameters $\alpha=1.5$, $\alpha=0.5$, and $\alpha=0.25$, while
keeping the mean service time fixed at $\E[S_1]=1/\mu_1$. The case
$\alpha=1.5$ has an increasing hazard rate, while $\alpha=0.5$ and
$\alpha=0.25$ have decreasing hazard rates. All other parameters and
policies are the same as in the benchmark $\alpha=1$ case.

For $\alpha=1.5$, we observe the same behavior as in the exponential benchmark, as shown in \Cref{fig:weibull150}. The $(k,n)$-convoy policies and their convex combinations continue to outperform all of the heuristic policies considered. We believe our optimality result should also hold in the case where small jobs have distributions with decreasing hazard rate, we leave this for future work.

\begin{figure}[h]
    \centering
    \includegraphics[width=0.8\textwidth]{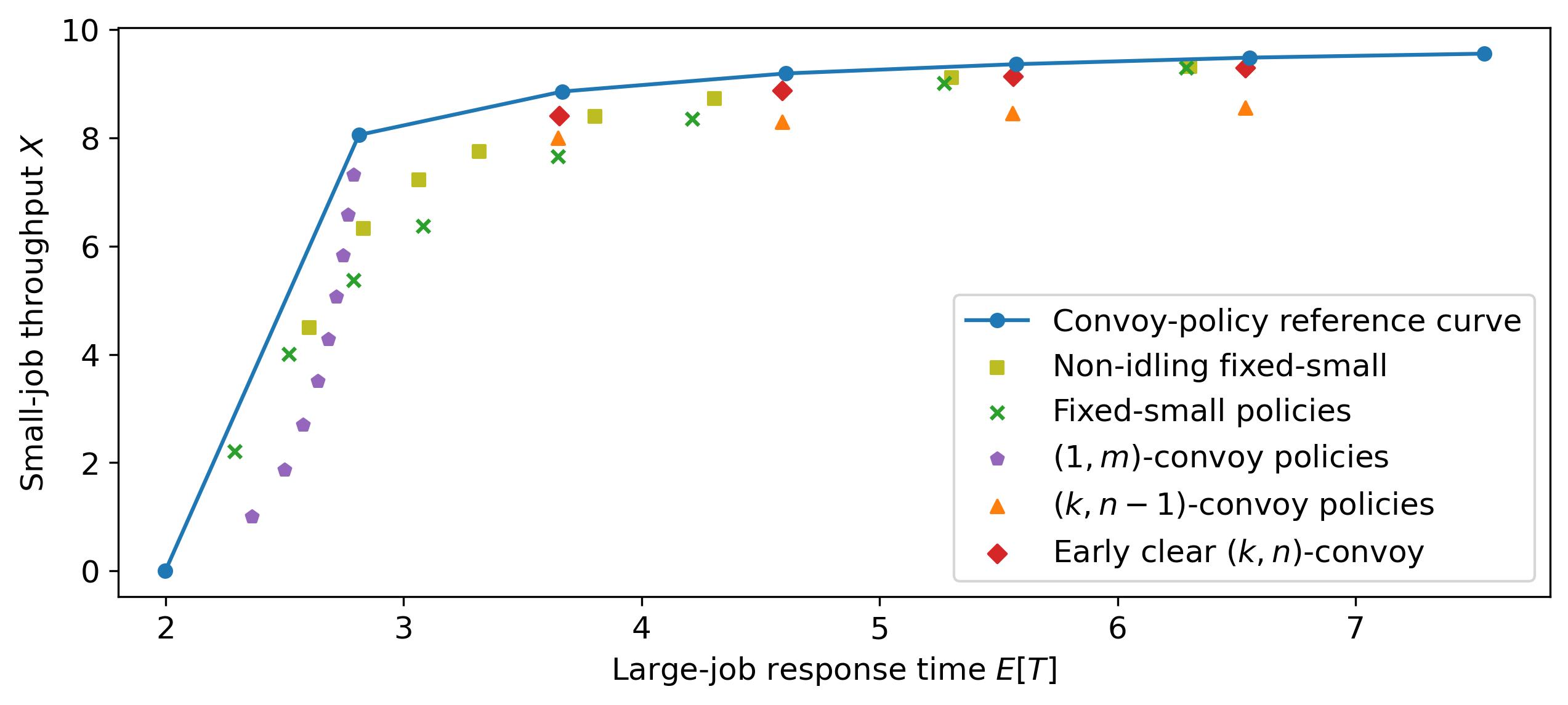}
    \caption{Small-job service times are Weibull with shape parameter
    $\alpha=1.5$ and the same mean $1/\mu_1$ as in
    \Cref{fig:weibull100}, this distribution has an increasing hazard rate. All other parameters and policies are unchanged from \Cref{fig:weibull100}.}
    \label{fig:weibull150}
\end{figure}

We next consider decreasing-hazard-rate service times.
\Cref{fig:weibull50} shows the results for $\alpha=0.5$. Although this
case is outside the assumptions of our theoretical results, the optimality structure of the benchmark case is preserved. In particular, the $(k,n)$-convoy policies and their convex combinations continue to outperform the other policies considered.

\begin{figure}[h]
    \centering
    \includegraphics[width=0.8\textwidth]{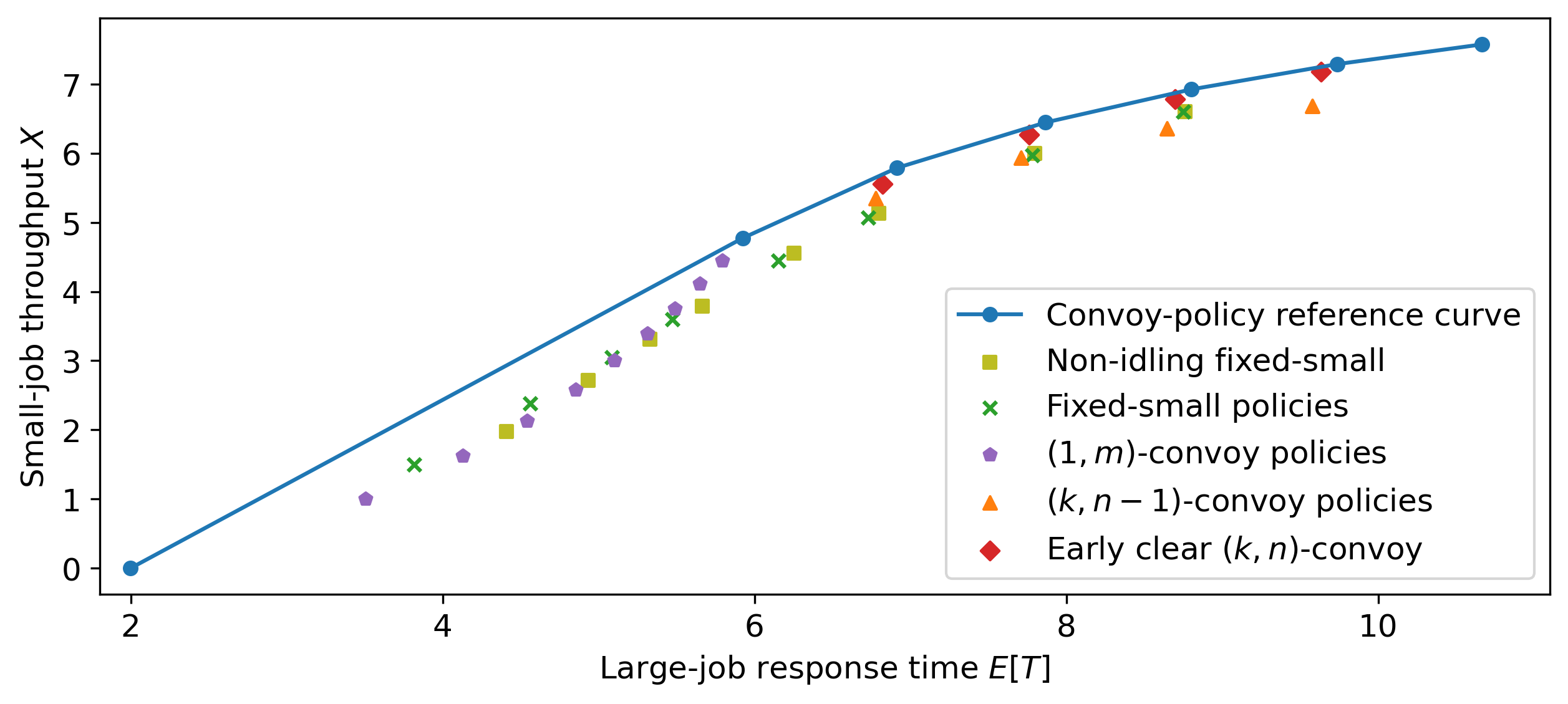}
    \caption{Small-job service times are Weibull with shape parameter $\alpha=0.5$ and the same mean $1/\mu_1$ as in \Cref{fig:weibull100}, this distribution has a decreasing hazard rate. All other parameters and policies are unchanged from \Cref{fig:weibull100}.}
    \label{fig:weibull50}
\end{figure}

The behavior changes more substantially when the shape parameter is reduced
to $\alpha=0.25$, as shown in \Cref{fig:weibull25}. In this case, some
of the early-clearing policies, as well as the $(1,n-1)$-convoy policy, can
outperform the curved formed by convex combinations of the
$(k,n)$-convoy policies. Because this curve may no longer be the optimal frontier, we call it the \textit{reference curve}. Thus, the optimality structure established under
exponential small-job service times need not persist when the small-job
distribution has a sufficiently strongly decreasing hazard rate. In such cases, policies that anticipate large-job arrivals and preemptively make servers available for them may achieve better performance.

\begin{figure}[h]
    \centering
    \includegraphics[width=0.8\textwidth]{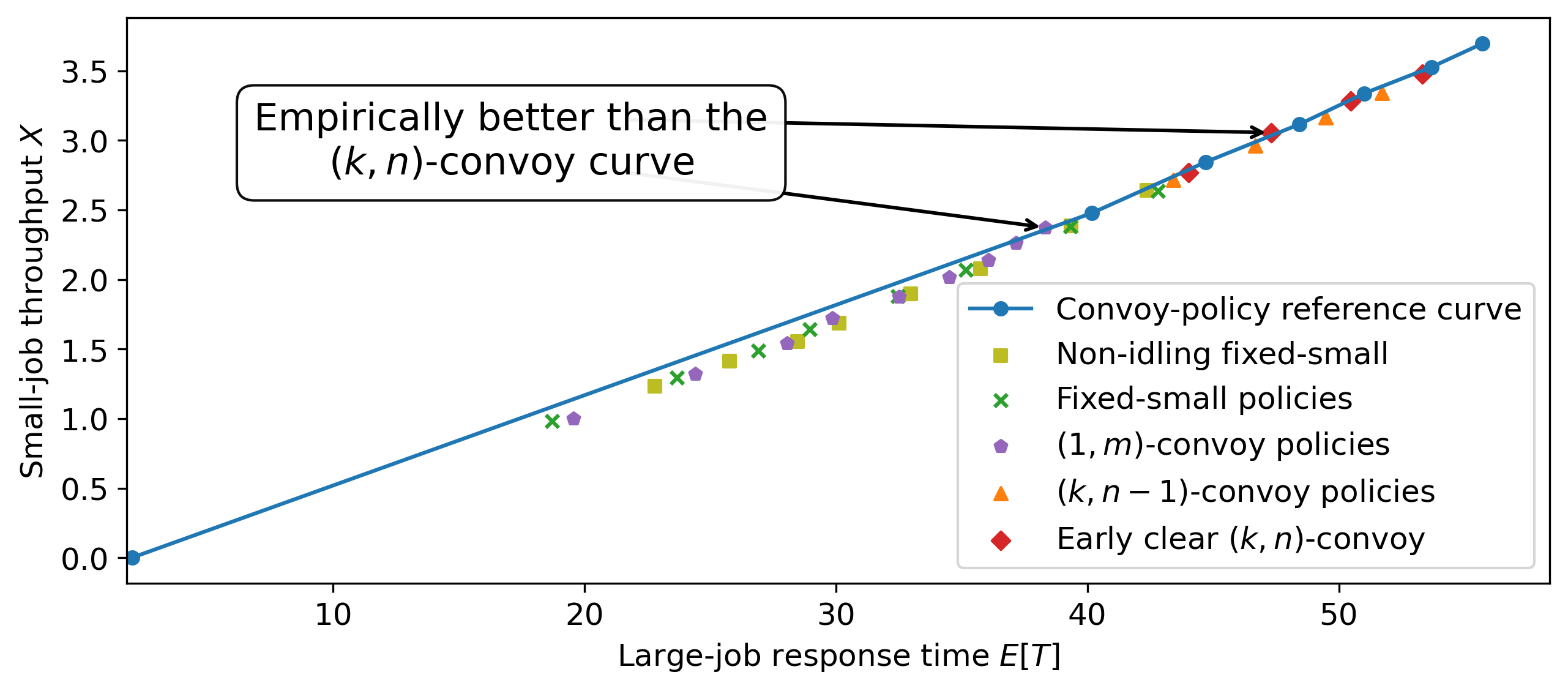}
    \caption{Small-job service times are Weibull with shape parameter $\alpha=0.25$ and the same mean $1/\mu_1$ as in \Cref{fig:weibull100}, this distribution has a decreasing hazard rate. All other parameters and policies are unchanged from \Cref{fig:weibull100}.}
    \label{fig:weibull25}
\end{figure}

A key intuition behind this change is the behavior of the largest small-job service time. As $\alpha$ decreases, increasing $n$ makes waiting for all simultaneously admitted small jobs to clear increasingly costly for some nonexponential distribution. In the exponential case, the additional penalty from the largest of a group of small jobs grows only as $\log n$, and this clearing cost is concave in $n$. Thus, adding more small jobs improves throughput while incurring a diminishing additional clearing penalty, which motivates the $(k,n)$-convoy policies.

For nonexponential service times, this need no longer be the case. Let $M_n=\max\{S_{1,1},\ldots,S_{1,n}\}$ for $n$ i.i.d.\ Weibull small-job service times. For shape parameter $\alpha$ and scale parameter $\theta$, $\E[M_n]\sim \theta(\log n)^{1/\alpha}$ as $n\to\infty.$
Hence, for $\alpha=1$ the maximum grows as $\log n$, while for $\alpha=0.5$ and $\alpha=0.25$ the asymptotic growth rates become $(\log n)^2$ and $(\log n)^4$, respectively.

For $\alpha=0.25$, the $(\log n)^4$ scaling grows particularly rapidly over the range $n\in[1,10]$, making early clearing more beneficial. This is reflected in \Cref{fig:weibull25}, where the $(1,n-1)$-convoy policy and some early-clearing policies lie above the reference curve formed by the $(k,n)$-convoy policies and their convex combinations. Additionally, this reference curve is no longer visibly strictly concave, it appears nearly linear and empirically nonconcave in some regions. However, for large $k$ this small-job throughput will eventually be constant, even though it is locally nearly linear for small $k$.

\subsection{Large jobs using fewer servers}

We next relax the assumption that a large job occupies all $n$ servers. Suppose instead that each large job requires $b<n$ servers simultaneously. We first consider the case $b>n/2$. In this regime, at most one large job can be served at a time, so the basic structure of the original model is preserved. 

We therefore expect the $(k,n)$-convoy policies and their convex combinations to retain the same optimality structure. The main difference is that, while a large job is in service, the remaining $n-b$ servers can continue serving small jobs. Consequently, the leftmost point of the candidate Pareto frontier is now the $(1,n-b)$-convoy policy rather than the $(1,0)$-convoy policy.

The $(k,n)$-convoy policy extends naturally to this setting. The system serves small jobs on all $n$ servers until the $k$th large-job arrival, at which point it clears only $b$ small jobs to make room for large-job service, while the remaining $n-b$ servers continue serving small jobs. The $b$ servers used for large-job service may vary across renewal cycles. The other benchmark policies are extended analogously.

We simulate the case $n=10$ and $b=8$, with exponential small-job distribution. All other parameters are unchanged. As shown in \Cref{fig:n8}, the $(k,n)$-convoy policies and their convex combinations form the candidate Pareto frontier, beginning at the $(1,n-b)=(1,2)$-convoy policy. All policies are simulated over 5 million renewal cycles. 

We observe that all benchmark policies lie below the reference curve and exhibit the same behavior as in the original all-server setting, with the convoy policies and their convex combinations outperforming all benchmark policies. Additionally, we empirically observe a larger gap between the optimal frontier and some of the heuristic policies, such as the fixed-small policies, compared to the exponential benchmark case.

\begin{figure}[h]
    \centering
    \includegraphics[width=0.8\textwidth]{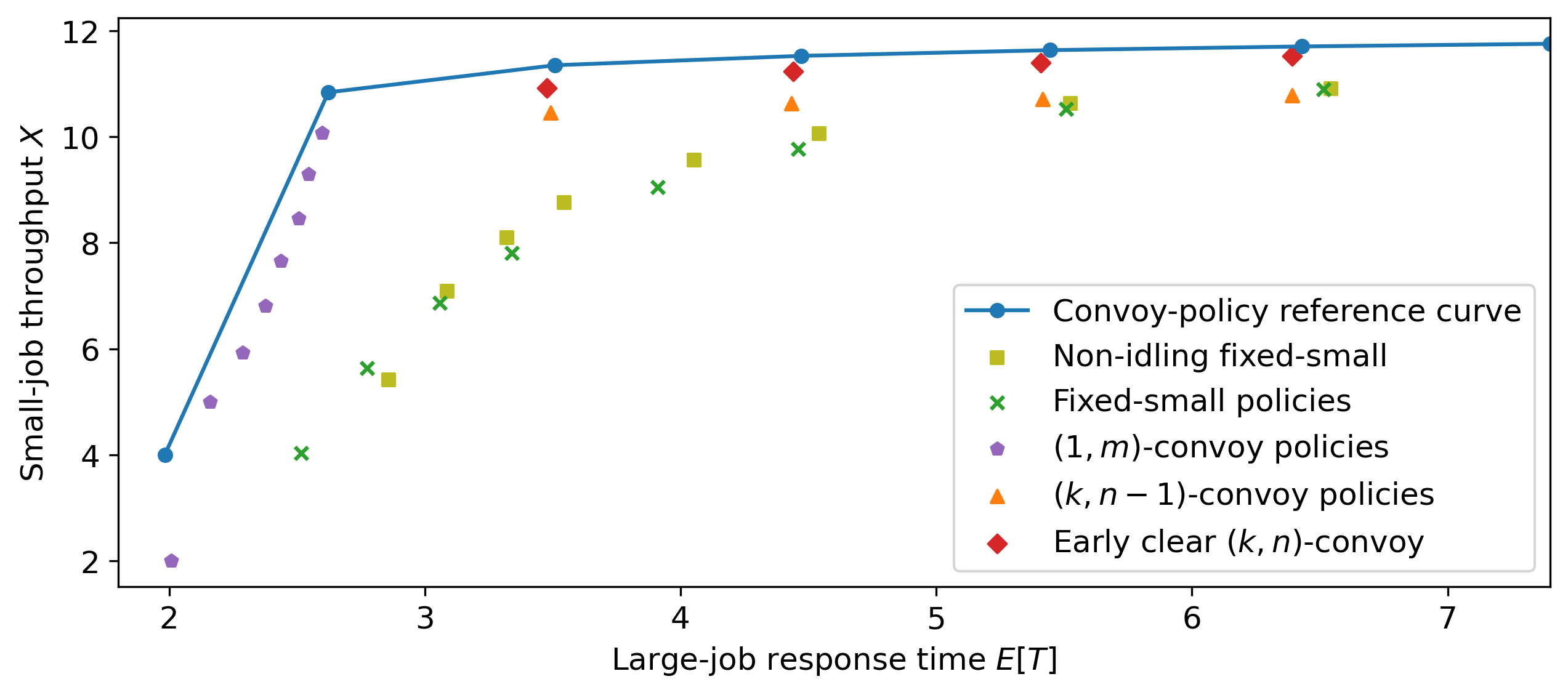}
    \caption{Large jobs require $b=8$ of the $n=10$ servers. Small-job service times are exponential, and all other parameters are unchanged. All policies are simulated over 5 million renewal cycles.}
    \label{fig:n8}
\end{figure}

The case $b\leq n/2$ is more complicated. In this regime, two or more large jobs can be served simultaneously. We therefore do not expect a direct analogue of the simple $(k,n)$-convoy family to characterize the entire Pareto frontier. In particular, the system must decide when to admit more than one large job. We expect this tradeoff to depend on the large-job arrival rate, among other parameters. Consequently, the class of potentially Pareto-optimal policies is substantially richer, and we leave a study of this regime for future work.

\section{Proofs}
\label{sec:proof}

In this section, we prove \Cref{thm:main}. We begin by establishing the required optimality lemmas, organized into three subsections corresponding to the three parts of the proof overview in \Cref{sec:overview}.

\subsection{Renewal cycles}
In this subsection, we derive explicit expressions for key performance metrics of the $(k,m)$-convoy policies, where the case $m=n$ will serve as our candidate family of Pareto-optimal policies. We then prove our first optimality result, namely, that any Pareto-optimal policy must serve large jobs exhaustively once large-job service begins.

Let $\E[T^{k,m}]$ denote the mean response time of a large job
under the $(k,m)$-convoy policy, $\E[L^{k,m}]$ the mean cycle
length, $\E[C^{k,m}]$ the mean total cumulative response time
of large jobs in a cycle, and $X^{k,m}$ the throughput of
small jobs.

\begin{restatable}{proposition}{propT}
\label{prop:kmconvoy}
Under the $(k,m)$-convoy policy:
\begin{itemize}
    \item The mean response time of large jobs is $\E[T^{k,m}]
        =
        \E[T^{M/G/1\text{-}L}]
        +
        \frac{k(k-1)}{2\lambda(k+\lambda\E[I_m])}
        +
        \frac{k\E[I_m]}{k+\lambda\E[I_m]}
        +
        \frac{\lambda\E[I_m^2]}{2(k+\lambda\E[I_m])}.$

    \item The mean cycle length is $\E[L^{k,m}]
        =
        \frac{k}{\lambda}
        +
        \frac{\E[I_m]+k}{1-\lambda}.$
    \item The throughput of small jobs is $X^{k,m}
        =
        \frac{m(\lambda+k\mu_1)(1-\lambda)}{k+\lambda\E[I_m]}.$
\end{itemize}
\end{restatable}

The proof of \Cref{prop:kmconvoy} is deferred to \ref{app:1}. In \Cref{prop:kmconvoy}, $I_m$ represents the setup cost, interpreted here as the time required to clear the $m$ small jobs before large-job service begins. The response time $\E[T^{k,m}]$ and cycle length $\E[L^{k,m}]$ formulas are a special case of a more general result about systems with setup costs.

An important property of $\E[T^{k,m}]$, $\E[L^{k,m}]$, and $\E[C^{k,m}]$ is that, for fixed $k$, they are increasing and concave functions of $m$. These results are proven in \ref{app:1}.

\begin{restatable}{lemma}{lemTconcave}
\label{lem:Tconcave}
    For fixed $k$, $\E[T^{k,m}]$ is increasing and concave in $m$.
\end{restatable}

\begin{restatable}{lemma}{lemLCconcave}
\label{lem:LCconcave}
    For fixed $k$, both $\E[L^{k,m}]$ and $\E[C^{k,m}]$ are increasing and concave in $m$.
\end{restatable}

Next, we show that connecting neighboring points of the $(k,n)$-convoy policies in the $(\E[T],X)$-plane forms a curve whose region below the curve is convex. This is important because any convex combination of non-neighboring points lies strictly below the graph and therefore cannot be Pareto optimal. The proof of this lemma is deferred to \ref{app:1}.

\begin{restatable}{lemma}{lemConHypo}
\label{lem:pareto_hypograph}
Fix $n\ge 2$. Let $\Gamma_n$ be the union of
\begin{enumerate}
    \item the line segment joining $\left(\E[T^{1,0}],0\right)$ and $\left(\E[T^{1,n}],X^{1,n}\right)$, and
    \item the piecewise-linear curve obtained by connecting neighboring points of $\left(\E[T^{k,n}],X^{k,n}\right)$ for $k\geq 1$.
\end{enumerate}
Then the hypograph of $\Gamma_n$ is convex.
\end{restatable}

Next, we prove that any policy that does not complete all large jobs in queue once large-job service has begun is suboptimal. This establishes the first step towards establishing of our main theorem, \Cref{thm:main}.

\begin{lemma}
\label{lem:exhaustion}
    Any Pareto-optimal policy for $(\E[T],X)$ is exhaustive with respect to large jobs. That is, once the system begins serving large jobs, it continues to do so until the large-job queue becomes empty.
\end{lemma}

\begin{proof}
We prove the result by an interchange argument. Suppose, toward a contradiction, that a Pareto-optimal policy $\pi_1$ is not exhaustive with respect to large jobs. Then, with positive probability, there exists a sample path on which $\pi_1$ begins serving large jobs at time $t_1$, serves large jobs continuously until time $t_2$, but leaves at least one large job in the queue at time $t_2$. Instead of continuing to serve large jobs, $\pi_1$ performs actions that do not involve serving large jobs over an interval $[t_2,t_4)$. At time $t_4$, policy $\pi_1$ begins serving a large job $J$, completing it at time $t_5$.

We construct another policy $\pi_2$ on the same arrival and service sample path. The policy $\pi_2$ agrees with $\pi_1$ up to time $t_2$. At time $t_2$, instead of performing the non-large-job actions that $\pi_1$ performs over $[t_2,t_4)$, policy $\pi_2$ immediately serves the same large job $J$. Let $t_3$ be the completion time of $J$ under $\pi_2$, so that $t_3-t_2=t_5-t_4$. After time $t_3$, policy $\pi_2$ performs exactly the same non-large-job actions that $\pi_1$ performs over $[t_2,t_4)$, in the same order. Therefore, at time $t_3+(t_4-t_2)=t_5,$ both policies have completed the same large job $J$ and have performed the same non-large-job actions. Hence the two systems are in the same state at time $t_5$. From time $t_5$ onward, let $\pi_2$ follow $\pi_1$ exactly.

Thus, over the interval $[t_2,t_5]$, the two policies complete the same small jobs and perform the same amount of non-large-job work. Therefore, the small-job throughput of $\pi_2$ is the same as that of $\pi_1$ over this sample-path segment. It is easy to see that, relative to $\pi_2$, policy $\pi_1$ keeps the large job $J$ in the system for an additional time $t_5-t_3=t_4-t_2>0$, and therefore incurs exactly $t_4-t_2$ extra cumulative large-job response time over this sample-path segment.

Now consider renewal cycles defined by returns to the renewal state. Under the positive recurrence to the renewal state, these cycles are finite almost surely. Applying the above interchange to the first non-exhaustive large-job service period in a renewal cycle gives a policy that has the same cycle length and completes the same small jobs as $\pi_1$, but has strictly smaller cumulative large-job response time on a set of sample paths with positive probability. If $\pi_2$ reaches the renewal state before the corresponding renewal time of $\pi_1$, we simply concatenate the resulting subcycles and compare rewards over the full cycle of $\pi_1$, this does not change the long-run averages.

Since throughput and response time are time-average ratios involving the number of completions, cumulative response time, and cycle length, $\pi_2$ has the same long-run small-job throughput as $\pi_1$ but strictly smaller mean large-job response time. Thus, $\pi_2$ Pareto dominates $\pi_1$ with respect to the performance pair $(\E[T],X)$, contradicting the Pareto optimality of $\pi_1$.
\end{proof}

\subsection{Interarrival replications}
\label{sec:replication_proofs}

In this subsection, our goal is to construct, for a single interarrival period, a probabilistic mixture of threshold policies that replicates the corresponding behavior of an arbitrary policy $\pi$. More precisely, we show how to match the arrival-seen distribution induced by $\pi$ between two consecutive large-job arrivals, where the arrival-seen distribution is the distribution of the small-job level observed by the next large-job arrival. Applying this construction successively over interarrival periods then yields an overall replication of $\pi$ by threshold policies over a renewal cycle.

For each level $r\ge 1$, let $\tau_r\sim\mathrm{Exp}(r\mu_1)$
and $A\sim\mathrm{Exp}(\lambda)$ be independent random variables.
Define
\begin{align*}
p_r
:= \mathbb{P}(\tau_r<A)
= \frac{r\mu_1}{r\mu_1+\lambda}\quad\text{and}\quad q_r
&:= 1-p_r
= \frac{\lambda}{r\mu_1+\lambda}
\end{align*}
Equivalently, $p_r$ is the probability that, given the system is at level $r$,
a small completion occurs before a large arrival. Define the products $\beta_r := \prod_{j=r+1}^{n} p_j$ for $r=0,1,\ldots,n$ with $\beta_n := 1.$

Equivalently, $\beta_r$ is the probability that,
starting from level $n$, the system reaches level $r$
through successive small completions before the next large arrival occurs, if no additional small jobs are added to service.

From the definition, for each $r=1,\ldots,n$, $\beta_{r-1}
= \prod_{j=r}^{n} p_j
= p_r \prod_{j=r+1}^{n} p_j
= p_r \beta_r.$
Hence
\begin{align}\label{eq:lem3_0_1}
\beta_r q_r
= \beta_r(1-p_r)
= \beta_r - \beta_{r-1}.
\end{align}

Ordering rows and columns by levels $\{n,n-1,\dots,0\}$, we define the matrix $P$
whose columns are the threshold distributions is
\begin{align}
\label{eqn:matrix}
P =
\begin{pmatrix}
\beta_n & \beta_n q_n & \beta_n q_n & \cdots & \beta_n q_n \\
0 & \beta_{n-1} & \beta_{n-1} q_{n-1} & \cdots & \beta_{n-1} q_{n-1} \\
0 & 0 & \beta_{n-2} & \cdots & \beta_{n-2} q_{n-2} \\
\vdots & & \ddots & \ddots & \vdots \\
0 & \cdots & 0 & 0 & \beta_0
\end{pmatrix}.
\end{align}

In particular, each column of $P$ is the arrival-seen distribution of the small-job level at a large-job (Poisson) arrival under the Drop-to-$\ell$ policy with $n$ initial small jobs in service, with levels ordered from $n$ down to $0$. The first column is also the Hold-at-$n$ policy.

Let $u^{\pi}_{(k)}=(u_n,\ldots,u_1,u_0)^\top$ denote the long-run distribution of the small-job level observed between the $(k-1)$-th and $k$-th large-job arrivals. For brevity, we write $u^\pi$. By PASTA, this long-run distribution coincides with the arrival-seen distribution of the $k$-th large job arrival.

The goal is to construct, at the $(k-1)$-th large-job arrival, a probabilistic mixture of threshold policies such that the resulting long-run distribution between the $(k-1)$-th and the $k$-th large-job arrival is exactly $u^\pi$. Specifically, we select among the threshold policies according to a probability vector $v^\pi$ satisfying
\begin{align*}
v^{\pi} = P^{-1}u^{\pi }.
\end{align*}

This vector $v^\pi$ defines a replication strategy between the $(k-1)$-th and $k$-th large-job arrivals. The decision rule of an arbitrary policy $\pi$ within the interval is irrelevant beyond the induced long-run distribution $u^\pi$, and the probabilistic mixture $v^\pi$ reproduces $u^\pi$ exactly. In particular, if we denote the resulting mixture of threshold policies by MoT-$\pi$, then $u^{\mathrm{MoT}\text{-}\pi}=u^\pi.$

This replication strategy is stated under the assumption that the initial level at the $(k-1)$-th arrival is $n$. The general replication strategy, which allows the level at the $(k-1)$-th arrival to be strictly less than $n$, is given in \Cref{thm:2_1}.

Moreover, since the arrival-seen distribution coincides with the long-run distribution, matching $u$ also matches performance metrics that depend only on the level of small jobs over the interval, namely, the expected number of small jobs in service seen by the $k$-th arrival and the expected number of small-job completions between the $(k-1)$-th and $k$-th large-job arrivals.

We then apply this replication strategy inductively. Once the policy has been replicated up to the $(k-1)$-th arrival, we use $v_{(k)}^{\pi}=P^{-1}u_{(k)}^{\pi}$ to replicate the behavior of $\pi$ up to the next large-job arrival. Repeating this over successive arrivals yields a replication of $\pi$ by threshold policies over an entire renewal cycle.

We now construct $P^{-1}$ explicitly. We propose a candidate for $P^{-1}$ and then verify that $P P^{-1}=I$.

\begin{lemma}
\label{lem:2_1}
The matrix $P$ defined in \Cref{eqn:matrix} has an upper-triangular inverse $P^{-1}$. Moreover, each columns of both $P$ and $P^{-1}$ sum to $1$.
\end{lemma}

\begin{proof}
    The inverse is given by
    \begin{align*}
    P^{-1} =
    \begin{pmatrix}
    \beta_n^{-1} & -\frac{q_n}{\beta_{n-1}} & -\frac{q_n}{\beta_{n-1}} & \cdots & -\frac{q_n}{\beta_{n-1}} \\
    0 & \beta_{n-1}^{-1} & -\frac{q_{n-1}}{\beta_{n-2}} & \cdots & -\frac{q_{n-1}}{\beta_{n-2}} \\
    0 & 0 & \beta_{n-2}^{-1} & \cdots & -\frac{q_{n-2}}{\beta_{n-3}} \\
    \vdots & & \ddots & \ddots & \vdots \\
    0 & \cdots & 0 & 0 & \beta_0^{-1}
    \end{pmatrix}.
    \end{align*}

    We verify that $(PP^{-1})_{r,s}=\mathbf{1}\{r=s\}$, where rows and columns are indexed by levels
    $r,s\in\{0,1,\dots,n\}$ ordered as $n,n-1,\dots,0$. We write $P$ and $P^{-1}$ more compactly,
    \begin{align*}
    P_{r,\ell}=
    \begin{cases}
    0, & r<\ell,\\
    \beta_r, & r=\ell,\\
    \beta_r q_r, & r>\ell,
    \end{cases}
    \quad
    (P^{-1})_{\ell,s}=
    \begin{cases}
    0, & \ell<s,\\
    \beta_s^{-1}, & \ell=s,\\
    -\frac{q_\ell}{\beta_{\ell-1}}, & \ell>s,
    \end{cases}
    \end{align*}
    and $\beta_{\ell-1}=p_\ell\beta_\ell$ for $\ell\ge1$. Since $P_{r,\ell}=0$ for $\ell>r$ and
    $(P^{-1})_{\ell,s}=0$ for $\ell<s$, we have $(PP^{-1})_{r,s}=\sum_{\ell=s}^{r} P_{r,\ell}(P^{-1})_{\ell,s}\quad (r\ge s).$
    
    If $r=s$, then only the term $\ell=r$ is nonzero, hence $(PP^{-1})_{r,r}=P_{r,r}(P^{-1})_{r,r}=\beta_r\cdot \beta_r^{-1}=1.$
    If $r>s$, then $P_{r,r}=\beta_r$ and $P_{r,\ell}=\beta_r q_r$ for all $\ell\le r-1$, while
    $(P^{-1})_{r,s}=-q_r/\beta_{r-1}$. Therefore, $(PP^{-1})_{r,s}
    =\beta_r q_r\sum_{\ell=s}^{r-1}(P^{-1})_{\ell,s}+\beta_r(P^{-1})_{r,s},$ for all $r>s.$
    
    It remains to compute the partial column sum. Using the explicit form of $P^{-1}$, $\sum_{\ell=s}^{r-1}(P^{-1})_{\ell,s}
    =\beta_s^{-1}-\sum_{\ell=s+1}^{r-1}\frac{q_\ell}{\beta_{\ell-1}}.$ Since $\beta_{\ell-1}=p_\ell\beta_\ell$, we have $\frac{q_\ell}{\beta_{\ell-1}}
    =\frac{q_\ell}{p_\ell\beta_\ell}
    =\left(\frac{1}{p_\ell}-1\right)\frac{1}{\beta_\ell}
    =\frac{1}{\beta_{\ell-1}}-\frac{1}{\beta_\ell},$
    and thus the sum telescopes $\sum_{\ell=s+1}^{r-1}\frac{q_\ell}{\beta_{\ell-1}}
    =\sum_{\ell=s+1}^{r-1}\left(\frac{1}{\beta_{\ell-1}}-\frac{1}{\beta_\ell}\right)
    =\frac{1}{\beta_s}-\frac{1}{\beta_{r-1}}.$
    
    Substituting back yields $\sum_{\ell=s}^{r-1}(P^{-1})_{\ell,s}
    =\beta_s^{-1}-\left(\beta_s^{-1}-\beta_{r-1}^{-1}\right)
    =\beta_{r-1}^{-1}.$
    Hence $(PP^{-1})_{r,s}
    =\beta_r q_r\,\beta_{r-1}^{-1}+\beta_r\left(-\frac{q_r}{\beta_{r-1}}\right)=0,$
    which proves $(PP^{-1})_{r,s}=\mathbf{1}\{r=s\}$ for $r\ge s$. Since $PP^{-1}$ is upper triangular,
    this shows that $PP^{-1}=I$, since for $r<s$, all terms are zero. 
\end{proof}

Next, we show that the entries of $v^\pi=P^{-1}u^\pi$ are all nonnegative, so that $v$ defines a valid probabilistic mixture.

\begin{lemma}
\label{lem:2_2}
    Let $u^\pi=(u_n,\ldots,u_1,u_0)^\top$ denote the arrival-seen distribution of the system level observed by the $k$-th Poisson arrival under an arbitrary policy $\pi$, assuming the system is at level $n$ at the instant immediately after the $(k-1)$-th Poisson arrival. Then the vector $v^\pi:=P^{-1}u^\pi$ satisfies $v^\pi_r \ge 0$ for all $r \in \{0,1,\ldots,n\}$.
\end{lemma}

\begin{proof}
    Let $u=(u_n,\cdots,u_1,u_0)^\top$ be the arrival-seen distribution under this policy. In this proof, we suppress the superscript $\pi$.

    Define the CDF at level $r$, $F_r:=\sum_{j=0}^r u_j=\mathbb{P}(R\leq r)$ for $r=0,1,...,n,$
    where $R$ is the level seen by the next large arrival. From the upper-triangular structure of $P^{-1}$, write the components of $v=P^{-1}u$ as
    \begin{align*}
         v_0=\frac{\pi_0}{\beta_0}=\frac{F_0}{\beta_0}\quad\text{and}\quad  v_r=\frac{\pi_r}{\beta_r}-\frac{q_r}{\beta_{r-1}}\sum_{j=0}^{r-1}\pi_j=\frac{\pi_r}{\beta_r}-\frac{q_r}{\beta_{r-1}}F_{r-1},\quad r=1,...,n.
    \end{align*}
    
    Using $\beta_{r-1}=p_r\beta_r$, this becomes
        \begin{align*}
            v_0=\frac{F_0}{\beta_0}\quad\text{and}\quad v_r=\frac{\pi_r-\frac{q_r}{p_r}F_{r-1}}{\beta_r}=\frac{F_r}{\beta_r}-\frac{F_{r-1}}{\beta_{r-1}},\quad r=1,...,n.
        \end{align*}
         As a sanity check, $            \sum_{r=0}^n v_r =\frac{F_n}{\beta_n}=1$
        since $F_n=1$ and $\beta_n=1$. 
        
        Next, to show that $v_r\geq 0$, observe that for each $r=1,...,n$, $F_{r-1}\leq p_r F_r.$
        This is true because $\mathbb{P}(R\leq r-1)=\mathbb{P}(R\leq r-1\mid R\leq r)\cdot\mathbb{P}(R\leq r)$ and $\mathbb{P}(R\leq r-1\mid R\leq r)\leq p_r$, where the equality holds only for the policy that drops from $r$ to $r-1$ at the exact first possible completion.

        Consequently, divide $F_{r-1}\leq p_rF_r$ by $\beta_{r-1}=p_r\beta_r$ we have  $\frac{F_{r-1}}{\beta_{r-1}}\leq \frac{F_r}{\beta_r}.$ Therefore, the scaled CDF sequence $\{F_r/\beta_r\}_{r=0}^n$ is nondecreasing, and its successive differences are nonnegative, i.e., $v_r=\frac{F_r}{\beta_r}-\frac{F_{r-1}}{\beta_{r-1}}\geq0$ for $r=1,...,n,$ and trivially $v_0=F_0/\beta_0\geq 0$.
\end{proof}

We have now handled the case where the level immediately after the $(k-1)$-th large-job arrival is $n$, and showed that $v^\pi=P^{-1}u^\pi$ defines a valid mixture of threshold policies.

The next theorem extends this construction to the general case where the initial level is $m\le n$.

\begin{theorem}
\label{thm:2_1}
    Let $u^\pi=(u_n,\ldots,u_1,u_0)^\top$ denote the arrival-seen
    distribution of the system level observed by the $k$-th
    large-job arrival under an arbitrary policy $\pi$. We can
    explicitly construct a convex combination of threshold
    policies whose behavior between the $(k-1)$-th and $k$-th
    large-job arrivals replicates that of $\pi$.
\end{theorem}

\begin{proof}
    Suppose that at the instant of the $(k-1)$-th Poisson arrival the system is at level $m$. We construct a mixture of threshold policies whose arrival-seen distribution at the $k$-th Poisson arrival matches $u=(u_n,\ldots,u_1,u_0)^\top$.
    In this proof, we suppress the superscript $\pi$.
    
    The construction is obtained by splitting the target distribution into two parts. For levels above $m$, we use Hold-at-$\ell$ policies with weights given directly by $u$. For levels at or below $m$, we use a mixture of Drop-to-$\ell$ policies. Specifically, for $\ell > m$, we use a probabilistic mixture of Hold-at-$\ell$ policies, choosing Hold-at-$\ell$ with probability $u_\ell$.
    
    On the other hand, for $\ell \leq m$, let $u_{\leq m} := (u_{m},\ldots,u_1,u_0)^\top$ be the subvector corresponding to levels below $m$. Consider the matrix $P^{(m)}$ in place of $P$ in \Cref{lem:2_1}, where $P^{(m)}$ is defined analogously to $P$ but on the reduced state space $\{m,m-1,\ldots,0\}$. By \Cref{lem:2_2}, the vector $v := (P^{(m)})^{-1}u_{\leq m}$
    exists and satisfies $v_\ell\ge 0$ for all $\ell\in\{m,m-1,...,0\}$. We then use a probabilistic mixture of Drop-to-$\ell$ (from level $m$) policies, choosing Drop-to-$\ell$ with probability $v_\ell$ for each $\ell\in\{m,m-1,\ldots,0\}$.

    By \Cref{lem:2_1}, the columns of $P^{(m)}$ and $(P^{(m)})^{-1}$ sum to $1$. So that $\sum_{\ell=0}^{m} v_\ell= \sum_{\ell=0}^{m} u_\ell,$
    and therefore the probabilities assigned for the case $\ell>m$ and $\ell\leq m$ sum to 1.
    
    Thus the resulting convex combination of simple policies replicates the arrival-seen distribution $u$ of the $k$-th Poisson arrivals.
\end{proof}

\Cref{thm:2_1} gives the general one-interarrival replication result. Applying it successively over interarrival intervals yields a full replication of an arbitrary policy by threshold policies.

\begin{definition}
\label{def:motpi}
Given a policy $\pi$, define MoT-$\pi$ as the mixture-of-threshold replication of $\pi$ as follows:
\begin{enumerate}
    \item Suppose MoT-$\pi$ has replicated $\pi$ up to the $(k-1)$-th large-job arrival, and let $m$ be the small-job level at that arrival. Let $u_{(k)}^\pi=(u_n,\ldots,u_1,u_0)^\top$ denote the arrival-seen distribution of the small-job level at the $k$th large-job arrival under $\pi$, conditional on the current state.
    \item During the interarrival period between the $(k-1)$-th and $k$-th large-job arrivals, MoT-$\pi$ selects a threshold policy according to the mixture constructed in \Cref{thm:2_1}. 
    \item This construction is repeated at each subsequent large-job arrival until the renewal cycle ends.
\end{enumerate}
\end{definition}

By construction, MoT-$\pi$ matches the arrival-seen small-job level distribution of $\pi$ over every interarrival period, and hence replicates $\pi$ over the entire renewal cycle. 

\begin{corollary}
\label{cor:2_1}
Under \Cref{def:motpi}, $\mathrm{MoT}\text{-}\pi$ replicates the long-run behavior of $\pi$ over the entire renewal cycle.
\end{corollary}

We conclude this subsection with a simple example illustrating the
one-interarrival replication construction. Let $n=2$ and $\lambda=\mu_1=1/2$. Consider a policy $\pi$
that follows the $(2,2)$-convoy policy, except for the following
modification between the first and second large-job arrivals.
Starting with two small jobs in service, $\pi$ replaces the first
small job that completes, but does not replace the second.
It then holds the system at level $1$ until the second large-job
arrival. If that arrival occurs before the second small-job
completion, the system is still at level $2$. At the second
large-job arrival, $\pi$ enters the clearing phase and subsequently
serves large jobs exhaustively.

Here, $p_2=2/3$ and $p_1=1/2$. With levels ordered as
$(2,1,0)$, the threshold matrix is
\begin{align*}
P=
\begin{pmatrix}
1 & 1/3 & 1/3\\
0 & 2/3 & 1/3\\
0 & 0   & 1/3
\end{pmatrix},
\end{align*}
whose columns correspond to Hold-at-$2$, Drop-to-$1$, and
Drop-to-$0$, respectively.

Under $\pi$, the second large-job arrival sees level $1$ exactly
when two small-job completions occur before that arrival.
Since the first completed job is replaced, the total small-job
completion rate remains $2\mu_1$ until the second completion.
Thus, $u^\pi=\begin{pmatrix}1-p_2^2, p_2^2, 0\end{pmatrix}^{\top}
=\begin{pmatrix}5/9, 4/9 ,0\end{pmatrix}^{\top},$ and $v^\pi=P^{-1}u^\pi=\begin{pmatrix}1/3, 2/3, 0\end{pmatrix}^{\top}.$

Thus, at the first large-job arrival, MoT-$\pi$ selects
Hold-at-$2$ with probability $1/3$ and Drop-to-$1$ with
probability $2/3$. The Hold-at-$2$ branch follows the
$(2,2)$-convoy policy over this interval. The Drop-to-$1$
branch clears one small job at the first completion
and then holds at level $1$. Its arrival-seen distribution is
$(1/3,2/3,0)^\top$, so the mixture gives $u^{\mathrm{MoT}\text{-}\pi}
=
\frac13\begin{pmatrix}1,0,0\end{pmatrix}^{\top}
+
\frac23\begin{pmatrix}1/3 , 2/3 , 0\end{pmatrix}^{\top}
=
\begin{pmatrix}5/9 , 4/9 , 0\end{pmatrix}^{\top}
=
u^\pi.$

\subsection{Suboptimality before clearing}
\label{sec:suboptimality_proof}

\Cref{cor:2_1} states that any policy $\pi$ can be replicated by a mixture of threshold policies, which we denote by MoT-$\pi$. In this section, we show that certain threshold-policy choices in MoT-$\pi$ are suboptimal.

Our focus is on the penultimate interarrival period of a renewal cycle, that is, the interarrival period immediately before the clearing phase, where a Drop-to-$0$ policy is used. Note that in a MoT-$\pi$ policy, the only threshold policy that can reach level $0$ is a Drop-to-$0$ policy selected at a large-job arrival, which by definition corresponds to the clearing phase.

The key idea is to improve any MoT-$\pi$ policy by modifying only the penultimate interarrival period before the clearing phase. Which interarrival period is penultimate may itself be random, since the policy may choose probabilistically whether to enter the clearing phase after the next large-job arrival. However, because the probability of each such branch can be computed in advance, we can condition on that branch and modify the corresponding penultimate interarrival period with the same probability. We then replace the original threshold-policy choice in that period by a different mixture of threshold policies that preserves $\E[N]$, where $\E[N]$ is the mean number of small jobs completed in a renewal cycle, while yielding strictly smaller large-job metrics.

During this penultimate interarrival period, MoT-$\pi$ must select one of the following threshold policies:
\begin{enumerate}
\item Hold-at-$m$: the policy holds the small-job level at $m$ until the next large-job arrival, after which the system enters the clearing phase.
\item Drop-to-$\ell$: the policy first decreases the small-job level from $m$ to $\ell$, then holds at level $\ell$ until the next large-job arrival, after which the system enters the clearing phase.
\end{enumerate}

We show that, with respect to the primary performance metrics $(X,\E[T])$, every threshold-policy choice in the penultimate interarrival period other than Drop-to-$0$ and Hold-at-$n$ is suboptimal. We will make use of the the auxiliary performance metrics $\E[N]$.

To carry this out, we consider any threshold-policy choice in the penultimate interarrival period other than Drop-to-$0$ and Hold-at-$n$, namely, Hold-at-$m$ with $m<n$ or Drop-to-$\ell$ with $\ell>0$, and compare it against mixtures of two extreme policies:
\begin{enumerate}
\item MoT-Full-$\pi$, which selects Hold-at-$n$ in the penultimate interarrival period and then Drop-to-$0$ in the subsequent clearing phase;
\item MoT-Clear-$\pi$, which selects Drop-to-$0$ already in the penultimate interarrival period, and in all subsequent periods.
\end{enumerate}
Both MoT-Full-$\pi$ and MoT-Clear-$\pi$ agree with MoT-$\pi$ up to the start of the final interarrival period before the clearing phase, and differ from MoT-$\pi$ only in their behavior during that penultimate interarrival period among each probability branch. Specifically, we explicitly construct a policy MoT-Mix-$\pi$, defined as a mixture of MoT-Full-$\pi$ and MoT-Clear-$\pi$ with a suitably chosen mixing probability $\alpha\in[0,1]$. With this choice of $\alpha$, MoT-Mix-$\pi$ matches policy $\pi$'s $\E[N]$ while achieving strictly smaller large-job metrics.

\subsubsection{Suboptimal Hold-at-$m$}

In this subsection, we show that selecting Hold-at-$m$ for some $m<n$ in the penultimate interarrival period is suboptimal. We distinguish two cases according to whether the penultimate interarrival period is also the first interarrival period of the renewal cycle.

If the penultimate interarrival period coincides with the first interarrival period of the renewal cycle, then the system starts at level $0$, so any Hold-at-$m$ policy can be selected freely; in particular, MoT-Clear-$\pi$ simply selects Hold-at-$0$. Otherwise, the initial level is determined by the preceding evolution of the renewal cycle, so the feasible threshold policies are constrained by that initial level. In particular, if the initial level is $\ell$, then Hold-at-$m$ can be chosen only for $m\ge \ell$.

We begin with the case where the penultimate interarrival period coincides with the first interarrival period of the renewal cycle, and compute the mixing coefficient $\alpha$ for MoT-Mix-$\pi$, defined as a mixture of MoT-Full-$\pi$ and MoT-Clear-$\pi$ so that $\E\left[N^{\text{MoT-Mix-}\pi}\right]=\E\left[N^\pi\right]$. Since these policies differ only in the penultimate interarrival period and the subsequent clearing phase, it suffices to match the expected number of small jobs completed over these two phases.

Let $U$ denote the number of small-job completions during the penultimate interarrival period, let $V$ denote the number of small jobs in service at the final large-job arrival, and define $W:=U+V$. Then $W$ is exactly the number of small jobs completed over the penultimate interarrival period and the subsequent clearing phase. Since MoT-$\pi$ replicates $\pi$, we have $\E\left[N^{\text{MoT-}\pi}\right]=\E\left[N^\pi\right]$, and hence matching $\E\left[W^{\text{MoT-Mix-}\pi}\right]=\E\left[W^{\text{MoT-}\pi}\right]$ implies $\E\left[N^{\text{MoT-Mix-}\pi}\right]=\E\left[N^\pi\right].$

\begin{lemma}
\label{lem:3_1}
Consider a policy MoT-$\pi$ that selects Hold-at-$m$ for some $m<n$ in the penultimate interarrival period, and suppose that this penultimate interarrival period is also the first interarrival period of the renewal cycle. Let MoT-Mix-$\pi$ denote the policy that selects MoT-Full-$\pi$ with probability $\alpha$ and MoT-Clear-$\pi$ with probability $1-\alpha$. If $\alpha=\frac{m}{n},$
then $\E\left[N^{\mathrm{MoT\text{-}Mix}\text{-}\pi}\right]=\E\left[N^\pi\right].$
\end{lemma}

\begin{proof}
For MoT-$\pi$, we have $V^{\text{MoT-}\pi}=m$ and $\E\left[W^{\text{MoT-}\pi}\right]
=
m+\frac{m\mu_1}{\lambda}.$ For MoT-Full-$\pi$, we have $V^{\text{MoT-Full-}\pi}=n$ and $\E\left[W^{\text{MoT-Full-}\pi}\right]
=
n+\frac{n\mu_1}{\lambda}.$ For MoT-Clear-$\pi$, because this is the first interarrival period of the renewal cycle, no small jobs are initially present and no small jobs are admitted, so $W^{\text{MoT-Clear-}\pi}=0.$

Therefore, $\E\left[W^{\text{MoT-Mix-}\pi}\right]
=
\alpha\left(n+\frac{n\mu_1}{\lambda}\right).$ Substituting $\alpha=\frac{m}{n}$ gives $\E\left[W^{\text{MoT-Mix-}\pi}\right]
=
\frac{m}{n}\left(n+\frac{n\mu_1}{\lambda}\right)
=
m+\frac{m\mu_1}{\lambda}
=
\E\left[W^{\text{MoT-}\pi}\right].$
Hence $\E\left[N^{\text{MoT-Mix-}\pi}\right]
=
\E\left[N^{\text{MoT-}\pi}\right]=\E\left[N^\pi\right].$
\end{proof}

Under the mixing coefficient in \Cref{lem:3_1}, $\E[N^{\text{MoT-Mix-}\pi}]$ matches $\E\left[N^{\text{MoT-}\pi}\right]$, and hence also $\E\left[N^\pi\right]$. We next show that under this same choice of $\alpha$, MoT-Mix-$\pi$ yields a smaller value of $\E\left[f(V)\right]$, where $V$ is the level at the start of the clearing phase and $f$ is any increasing discretely concave function. 

Here, $f$ represents a family of performance metrics that depend on $V$, the number of jobs present at the penultimate period, for which we will establish improvements under MoT-Mix-$\pi$. These include response time, cycle duration, and cumulative response time, each of which is shown in \ref{app:1} to be increasing and discretely concave.

\begin{lemma}
\label{lem:3_6}
Under the assumptions of \Cref{lem:3_1}, if $f:\{0,1,\dots,n\}\to\mathbb{R}$ is increasing and discretely concave, then $\E\left[f\left(V^{\text{MoT-Mix-}\pi}\right)\right]
\le
\E\left[f\left(V^{\text{MoT-}\pi}\right)\right].$
\end{lemma}

\begin{proof}
By construction of MoT-Mix-$\pi$ and the choice $\alpha=\frac{m}{n}$, we have $\E\left[f\left(V^{\text{MoT-Mix-}\pi}\right)\right]=\frac{m}{n}f(n)+\left(1-\frac{m}{n}\right)f(0)$. Because MoT-$\pi$ selects Hold-at-$m$ in the penultimate interarrival period, we also have $\E\left[f\left(V^{\text{MoT-}\pi}\right)\right]=f(m)$. Because $f$ is increasing and discretely concave, by definition of concavity, $f(m)\ge \frac{m}{n}f(n)+\left(1-\frac{m}{n}\right)f(0).$
Therefore, $\E\left[f\left(V^{\text{MoT-Mix-}\pi}\right)\right]\le \E\left[f\left(V^{\text{MoT-}\pi}\right)\right]$.
\end{proof}

Next, we consider the case where the penultimate interarrival period is not the first interarrival period of the renewal cycle. In this case, MoT-Clear-$\pi$ selects Drop-to-$0$.

\begin{lemma}
\label{lem:3_4}
Consider a policy MoT-$\pi$ whose level at the start of the penultimate interarrival period is $m<n$ and which selects Hold-at-$m$ in that period. Suppose that this penultimate interarrival period is not the first interarrival period of the renewal cycle. Let MoT-Mix-$\pi$ denote the policy that selects MoT-Full-$\pi$ with probability $\alpha$ and MoT-Clear-$\pi$ with probability $1-\alpha$. If $\alpha=\frac{m\mu_1}{n\mu_1+(n-m)\lambda},
$ then $\E\left[N^{\text{MoT-Mix-}\pi}\right]=\E\left[N^\pi\right].$
\end{lemma}

\begin{proof}
For MoT-$\pi$, we have $V^{\text{MoT-}\pi}=m$ and $\E\left[W^{\text{MoT-}\pi}\right]
=
m+\frac{m\mu_1}{\lambda}.$ For MoT-Full-$\pi$, we have $V^{\text{MoT-Full-}\pi}=n$ and $\E\left[W^{\text{MoT-Full-}\pi}\right]
=
n+\frac{n\mu_1}{\lambda}.$ For MoT-Clear-$\pi$, since the period starts at level $m$ and no additional small jobs are admitted, we have $W^{\text{MoT-Clear-}\pi}=m.$

Therefore, $\E\left[W^{\text{MoT-Mix-}\pi}\right]
=
\alpha\left(n+\frac{n\mu_1}{\lambda}\right)+(1-\alpha)m.$ Substituting $\alpha=\frac{m\mu_1}{n\mu_1+(n-m)\lambda}$ gives $\E\left[W^{\text{MoT-Mix-}\pi}\right]
=
m+\frac{m\mu_1}{\lambda}
=
\E\left[W^{\text{MoT-}\pi}\right].$
Hence $\E\left[N^{\text{MoT-Mix-}\pi}\right]
=
\E\left[N^{\text{MoT-}\pi}\right]
=
\E\left[N^\pi\right].$
\end{proof}

Again, under the mixing coefficient in \Cref{lem:3_4}, $\E[N^{\text{MoT-Mix-}\pi}]$ matches $\E\left[N^\pi\right]$. We next show that under this same choice of $\alpha$, MoT-Mix-$\pi$ yields a smaller value of $\E\left[f(V)\right]$.

\begin{lemma}
\label{lem:3_2}
Under the assumptions of \Cref{lem:3_4}, if $f:\{0,1,\dots,n\}\to\mathbb{R}$ is increasing and discretely concave, then $\E\left[f\left(V^{\text{MoT-Mix-}\pi}\right)\right]
\le
\E\left[f\left(V^{\text{MoT-}\pi}\right)\right].$
\end{lemma}

\begin{proof}
Because MoT-$\pi$ selects Hold-at-$m$, we have $\E\left[f\left(V^{\text{MoT-}\pi}\right)\right]=f(m).$ For MoT-Mix-$\pi$, let $\beta_r^{(m)}$ denote the truncated version of $\beta_r$, i.e., the probability that, starting from level $m$, the system reaches level $r$ before the next large-job arrival; see \ref{app:2} for the definition of $\beta_r^{(m)}$. By construction of MoT-Mix-$\pi$, $\E\left[f\left(V^{\text{MoT-Mix-}\pi}\right)\right]
=
\alpha f(n)+(1-\alpha)\left[\sum_{r=1}^{m}\beta_r^{(m)}q_r f(r)+\beta_0^{(m)}f(0)\right].$

The conclusion follows from \Cref{lem:b_2}. In particular, the proof of \Cref{lem:b_2} writes $f$ as the sum of its increments, computes the above difference from $f(m)$ directly, and uses the monotonicity of the increments of the discretely concave function $f$ to show that this difference is nonpositive.
\end{proof}

The inequalities in \Cref{lem:3_6,lem:3_2} depend on $f$ only through the assumptions of monotonicity and concavity, and the mixing coefficient $\alpha$ does not depend on the choice of $f$. In particular, the same bound and the same coefficient $\alpha$ apply to any function $f$ satisfying these properties.

\subsubsection{Suboptimal Drop-to-$\ell$}
We next rule out MoT-$\pi$ policies whose penultimate action is Drop-to-$\ell$. We only treat the case in which the initial level is $n$. If the initial level is $m$, then the same argument below shows that the MoT-$\pi$ policy that drops to $\ell$ from $m$ is dominated by a convex combination of MoT-Clear-$\pi$ (drop to $0$) and Hold-at-$m$. Because the previous subsection already showed that Hold-at-$m$ is suboptimal, it suffices to consider the case of initial level $n$.

\begin{lemma}
\label{lem:3_7}
Consider a policy MoT-$\pi$ whose level at the start of the penultimate interarrival period is $n$ and which selects Drop-to-$\ell$ in that period, where $\ell\in\{1,\dots,n-1\}$. Let MoT-Mix-$\pi$ denote the policy that selects MoT-Full-$\pi$ with probability $\alpha$ and MoT-Clear-$\pi$ with probability $1-\alpha$. If $\alpha=\frac{\ell}{n}\prod_{r=\ell+1}^n\frac{r\mu_1}{r\mu_1+\lambda}=\frac{\ell}{n}\beta_{\ell},$
then $\E\left[N^{\text{MoT-Mix-}\pi}\right]=\E\left[N^\pi\right].$
\end{lemma}

\begin{proof}
Recall that $\beta_r$ is the probability of reaching level $r$ before the next large-job arrival.
For MoT-Clear-$\pi$, no small jobs are admitted, so $W^{\text{MoT-Clear-}\pi}=n.$ For MoT-Full-$\pi$, we have $V^{\text{MoT-Full-}\pi}=n$, and hence $\E\left[W^{\text{MoT-Full-}\pi}\right]
=
n+\frac{n\mu_1}{\lambda}.$ For MoT-$\pi$, the initial $n$ small jobs contribute $n$ to $W$, and on the event of reaching level $\ell$ before the next large-job arrival, holding at level $\ell$ contributes an additional expected $\ell\mu_1/\lambda$. Therefore, $\E\left[W^{\text{MoT-}\pi}\right]
=
n+\beta_\ell\frac{\ell\mu_1}{\lambda}.$

By definition, $\E\left[W^{\text{MoT-Mix-}\pi}\right]
=
\alpha \E\left[W^{\text{MoT-Full-}\pi}\right]
+(1-\alpha)\E\left[W^{\text{MoT-Clear-}\pi}\right]=n+\alpha\frac{n\mu_1}{\lambda}.$
Substituting $\alpha=\frac{\ell}{n}\beta_\ell$ gives $\E\left[W^{\text{MoT-Mix-}\pi}\right]
=
n+\beta_\ell\frac{\ell\mu_1}{\lambda}
=
\E\left[W^{\text{MoT-}\pi}\right]$ and  $\E\left[N^{\text{MoT-Mix-}\pi}\right]
=
\E\left[N^{\text{MoT-}\pi}\right]
=
\E\left[N^\pi\right].$
\end{proof}

Under the mixing coefficient in \Cref{lem:3_7}, $\E[N^{\text{MoT-Mix-}\pi}]$ matches $\E\left[N^\pi\right]$. We next show that, under this same choice of $\alpha$, MoT-Mix-$\pi$ also yields a smaller value of $\E\left[f\left(V\right)\right]$ for any appropriate $f$.

\begin{lemma}
\label{lem:3_8}
Under the assumptions of \Cref{lem:3_7}, if $f:\{0,1,\dots,n\}\to\mathbb{R}$ is increasing and discretely concave, then $\E\left[f\!\left(V^{\text{MoT-Mix-}\pi}\right)\right]
\le
\E\left[f\!\left(V^{\text{MoT-}\pi}\right)\right].$
\end{lemma}

\begin{proof}
By the definitions of $\beta_r$ and $q_r$, we have $\E\left[f\!\left(V^{\text{MoT-}\pi}\right)\right]
=
\sum_{r=\ell+1}^n \beta_r q_r f(r)+\beta_\ell f(\ell).$
Also,
\begin{align*}
\E\left[f\!\left(V^{\text{MoT-Mix-}\pi}\right)\right]
&=
\alpha \E\left[f\!\left(V^{\text{MoT-Full-}\pi}\right)\right]
+(1-\alpha)\E\left[f\!\left(V^{\text{MoT-Clear-}\pi}\right)\right]\\
&=
\alpha f(n)
+(1-\alpha)\left[\sum_{r=1}^n \beta_r q_r f(r)+\beta_0 f(0)\right].
\end{align*}
With $\alpha=\frac{\ell}{n}\beta_\ell$ as in \Cref{lem:3_7}, the desired inequality is exactly the statement of \Cref{lem:b_3}.  Again, the proof of \Cref{lem:b_3} writes $f$ as the sum of its increments, computes the above difference from $f(m)$ directly, and uses the monotonicity of the increments of the discretely concave function $f$ to show that this difference is nonpositive.
\end{proof}

Thus, the MoT-$\pi$ policy that selects Drop-to-$\ell$ from level $n$ is dominated by a convex combination of MoT-Full-$\pi$ and MoT-Clear-$\pi$.

Similarly, the inequality in \Cref{lem:3_8} depends on $f$ only through the assumptions of monotonicity and concavity, and the mixing coefficient $\alpha$ is independent of the choice of $f$.

With the preceding lemmas in hand, we can now define MoT-Mix-$\pi$ over an entire renewal cycle. For each branch of the renewal-cycle evolution under MoT-$\pi$ that enters the clearing phase after the next large-job arrival, we step back to the corresponding penultimate interarrival period and replace the threshold-policy choice used on that branch by the improved mixture identified above. Because the probability of each such branch can be computed in advance, this replacement can be carried out with the corresponding probability preemptively, even if the original decision at a large-job arrival is not explicitly probabilistic.

\begin{definition}
\label{def:motmix}
Given a policy MoT-$\pi$, define MoT-Mix-$\pi$ as follows.
\begin{enumerate}
    
    \item On each branch where the penultimate interarrival period is also the first interarrival period of the renewal cycle and MoT-$\pi$ selects Hold-at-$m$ for some $m<n$, MoT-Mix-$\pi$ replaces that choice by the corresponding mixture of MoT-Full-$\pi$ and MoT-Clear-$\pi$ from \Cref{lem:3_1}, namely, with mixing coefficient $\alpha=m/n$.
    
    \item On each branch where the penultimate interarrival period is not the first interarrival period of the renewal cycle and MoT-$\pi$ selects Hold-at-$m$ for some $m<n$, MoT-Mix-$\pi$ replaces that choice by the corresponding mixture of MoT-Full-$\pi$ and MoT-Clear-$\pi$ from \Cref{lem:3_4}, namely, with mixing coefficient $\alpha=(m\mu_1)/(n\mu_1+(n-m)\lambda)$.
    
    \item On each branch where the penultimate interarrival period starts from level $n$ and MoT-$\pi$ selects Drop-to-$\ell$ for some $\ell\in\{1,\ldots,n-1\}$, MoT-Mix-$\pi$ replaces that choice by the corresponding mixture of MoT-Full-$\pi$ and MoT-Clear-$\pi$ from \Cref{lem:3_7}, namely, with mixing coefficient $\alpha=\ell\beta_\ell/n$.
    
    \item If the penultimate threshold-policy choice of MoT-$\pi$ is already Hold-at-$n$ or Drop-to-$0$, then MoT-Mix-$\pi$ leaves that branch unchanged.
\end{enumerate}
\end{definition}

\subsubsection{Optimality}

In this subsection, we show that MoT-Mix-$\pi$ performs at
least as well as MoT-$\pi$ in both $X$ and $\E[T]$, with at least strictly better performance in one of the metrics.

\begin{lemma}
\label{lem:3_9}
Under \Cref{def:motmix}, MoT-Mix-$\pi$ satisfies $X^{\text{MoT-Mix-}\pi}
\geq X^{\text{MoT-}\pi}$ and $\E\left[T^{\text{MoT-Mix-}\pi}\right]
\leq \E\left[T^{\text{MoT-}\pi}\right]$ with at least
one of these inequalities is strict.
\end{lemma}
\begin{proof}
In \Cref{lem:3_1,lem:3_4,lem:3_7}, the mixing coefficients are chosen so that $\E\left[N^{\text{MoT-Mix-}\pi}\right]
=
\E\left[N^{\text{MoT-}\pi}\right]
=
\E\left[N^\pi\right].$
Next, in \Cref{lem:3_6,lem:3_2,lem:3_8}, taking $f$ to be the cycle length $L$ of a renewal cycle, which is increasing and discretely concave in $V$ by \Cref{lem:LCconcave}, gives $\E\left[L^{\text{MoT-Mix-}\pi}\right]
\leq
\E\left[L^{\text{MoT-}\pi}\right].$ Because $X=\E[N]/\E[L]$, it follows that $X^{\text{MoT-Mix-}\pi}
=
\frac{\E\left[N^{\text{MoT-Mix-}\pi}\right]}{\E\left[L^{\text{MoT-Mix-}\pi}\right]}
\geq
\frac{\E\left[N^{\text{MoT-}\pi}\right]}{\E\left[L^{\text{MoT-}\pi}\right]}
=
X^{\text{MoT-}\pi}.$

For response time, note that $\E[T]=\E[C]/\left(\lambda \E[L]\right)$ is itself an increasing and discretely concave function of $V$, see \Cref{lem:Tconcave}. Therefore, by \Cref{lem:3_6,lem:3_2,lem:3_8}, taking $f$ to be the mean response time as a function of $V$, yields $\E\left[T^{\text{MoT-Mix-}\pi}\right]\leq\E\left[T^{\text{MoT-}\pi}\right].$
\end{proof}

\subsection{Proof of \Cref{thm:main}}
\label{sec:proof_main}
In this subsection, we prove \Cref{thm:main}. The proof follows the same structure as the proof overview in \Cref{sec:overview}, but presents the argument more succinctly.

\begin{proof}[Proof of \Cref{thm:main}]
    Let $\pi$ be an arbitrary policy that is positive recurrent to the renewal state. By \Cref{lem:exhaustion}, once $\pi$ begins serving large jobs, it must continue until all large jobs in the system have completed before admitting small jobs again.
    
    Next, by \Cref{thm:2_1} and \Cref{cor:2_1}, there exists a MoT-$\pi$ policy that replicates $\pi$ over successive interarrival periods by selecting appropriate threshold policies. Specifically, \Cref{thm:2_1} gives the explicit construction of the required convex combinations of threshold policies for a single interarrival period, and \Cref{cor:2_1} shows that applying this construction successively replicates $\pi$ over the entire renewal cycle.

    We then replace MoT-$\pi$ by the improved policy MoT-Mix-$\pi$ defined in \Cref{def:motmix}. By \Cref{lem:3_9}, MoT-Mix-$\pi$ achieves the better throughput smaller large-job response time than MoT-$\pi$ and hence $\pi$ unless the penultimate choice already has the Hold-at-$n$ form. Moreover, once the policy reaches level $0$, it must begin serving large jobs immediately.

    It remains to rule out any pre-penultimate choice other than Hold-at-$n$. If MoT-Mix-$\pi$ ever lowers the small-job level before the penultimate interarrival period, then, since the penultimate level under MoT-Mix-$\pi$ is $n$, the policy must later increase the small-job level at some large-job arrival. We call this behavior \emph{reneging}. We now show that reneging is suboptimal by constructing a non-reneging policy, denoted NR-MoT-Mix-$\pi$, that is at least as good as MoT-Mix-$\pi$.

    The construction is as follows. Suppose MoT-Mix-$\pi$ chooses Drop-to-$m$ with probability $p_m$, and conditional on this branch later increases the small-job level to $n$ with probability $p_{mn}$. We reallocate the probability mass $p_m p_{mn}$ from the Drop-to-$m$ branch to the Hold-at-$n$ branch, setting $\tilde p_m=p_m-p_m p_{mn}$ and $\tilde p_n=p_n+p_m p_{mn}$. The same reallocation is applied to any Hold-at-$m$ branch with $m<n$ that later increases to $n$. 
    
    This transformation preserves the probability of every post-arrival small-job level, so the large-job response-time contribution is unchanged. At the same time, the expected number of small-job completions is strictly larger whenever reneging occurs with positive probability. Applying this reallocation recursively, moving backward from the penultimate interarrival period to the start of the renewal cycle, eliminates all reneging. Hence MoT-Mix-$\pi$ is dominated by the non-reneging policy NR-MoT-Mix-$\pi$ whenever it reneges with positive probability.

    Thus, by the no-reneging property, MoT-Mix-$\pi$ cannot involve an increase at a large arrival followed by a decrease in an earlier interarrival period. Therefore, we conclude that any Pareto-optimal policy must always hold-at-$n$ until clearing phase or it is a convex combination of many such policies. In addition, \Cref{lem:pareto_hypograph} shows that any convex combination of non-adjacent convoy policies lies strictly below the Pareto curve and is therefore suboptimal. Hence only convex combinations of adjacent convoy policies can be Pareto optimal, completing the proof of \Cref{thm:main}.
\end{proof}

\section{Conclusion}

In this paper, we introduced the half-batch MSJ model, a novel hybrid open and closed model in the multiserver-job setting. In this model, large jobs arrive according to a Poisson process and require all servers, while small jobs are always available. This is motivated by real-world systems in which high-priority large jobs coexist with an many time-insensitive small jobs, creating a trade-off between large-job response time and small-job throughput.

We gave an exact characterization of this trade-off. We proved that the Pareto frontier is generated by the $(1,0)$-convoy policy and the family of $(k,n)$-convoy policies, together with convex combinations of neighboring frontier policies. The result is fully non-asymptotic and holds for every stable large-job arrival rate, every number of servers, and an arbitrary large-job size distribution.

We also use simulation to examine the robustness of our policies beyond the assumptions of the theoretical model. The results show that the optimality structure persists under moderate departures from these assumptions, while sufficiently heavy-tailed small-job sizes can make anticipatory policies more effective.

There are several possible future directions. One natural generalization is to allow large jobs to occupy only a fraction of the system, such as half of the servers. This extension is already complicated, and the Pareto-optimal family of policies may depend on the arrival rate, unlike our result. A further generalization in this direction is to allow multiple large-job classes with different server requirements, which is closer to practical heterogeneous MSJ systems. A different direction is to incorporate preemption overhead or cancellation costs for small jobs, although the latter may naturally lead to a three-dimensional performance trade-off.

\bibliographystyle{elsarticle-num-names} 
\bibliography{reference.bib}


\appendix

\numberwithin{lemma}{section} 
\renewcommand{\thelemma}{\Alph{section}.\arabic{lemma}} 
\renewcommand{\theHlemma}{appendix.\Alph{section}.\arabic{lemma}} 
\renewcommand{\theHequation}{appendix.\Alph{section}.\arabic{equation}}

\section{$(k,m)$-convoy policy}
\label{app:1}

\propT*

\begin{proof}
Fix $k$ and $m$. Observe that the $(k,m)$-convoy policy is equivalent to a $k$-policy with setup cost.

When the $k$th job arrives (starting from an empty system), a setup period of random length $I_m$ begins, after which the server processes jobs in FCFS order until the system becomes empty again. Thus the system's renewal cycle begins when all large jobs have completed and ends the next time $k$ jobs have accumulated in the queue.

Throughout, we normalize the service times of large jobs so that $\E[S]=1$, and $\rho=\lambda<1$.

\textbf{Priority-based policy.}
Consider a cycle and label the $k$ jobs present at the instant the setup begins by
\begin{align*}
\text{Job}_1,\text{Job}_2,\ldots,\text{Job}_k
\end{align*}
in order of arrival. 

Consider the following auxiliary priority rule: While $\text{Job}_1$ is in setup or service, any arriving job immediately gains priority over $\text{Job}_2$. Moreover, any arrivals that occur during the entire sub-busy period generated by $\text{Job}_1$ also have priority over $\text{Job}_2$ and are served in FCFS order. Once this sub-busy period completes and no jobs remain with higher priority than $\text{Job}_2$, service of $\text{Job}_2$ begins. The same rule is then applied recursively: arrivals during the sub-busy period generated by $\text{Job}_2$ have priority over $\text{Job}_3$, and so on, until $\text{Job}_k$ completes. This partitions the busy period into $k$ consecutive sub-busy periods:
\begin{align*}
B_1 \ \text{(with setup $I_m$)}\quad\text{and}\quad B_2,\ldots,B_k \ \text{(standard $M/G/1$ sub-busy periods)}.
\end{align*}
Because this auxiliary rule is non-preemptive and does not use job sizes, and because it has the same setup cost as the original FCFS queue within each cycle, it yields the same mean response time as the $(k,m)$-convoy policy \cite[Theorem~29.2]{Harchol-Balter2013}.

\textbf{Mean cycle length.}
Let $L^{k,m}$ be the cycle length. The idle portion is the time required for $k$ Poisson arrivals, denoted by $A_k$, therefore $\E[A_k]=\frac{k}{\lambda}.$

For the busy portion, $B_1$ is a busy period of an $M/G/1$ queue with setup $I_m$ started by one job, whose mean length is
\begin{align*}
\E[B_1]=\frac{\E[I_m]+\E[S]}{1-\rho}=\frac{\E[I_m]+1}{1-\lambda},
\end{align*}
and each $B_i$ for $i\ge 2$ is a standard $M/G/1$ busy period started by one job, with mean length $\E[B_i]=\frac{\E[S]}{1-\rho}=\frac{1}{1-\lambda}.$
Therefore, $\E[L^{k,m}]
=\E[A_k]+\E[B_1]+\sum_{i=2}^{k}\E[B_i] =\frac{k}{\lambda}+\frac{\E[I_m]+1}{1-\lambda}+\frac{k-1}{1-\lambda}
=\frac{k}{\lambda}+\frac{\E[I_m]+k}{1-\lambda}.$

\textbf{Mean response time.}
Let $T^{k,m}$ be the stationary response time. By PASTA, an arriving job falls in the interval $A_k$, the first sub-busy period $B_1$, or one of the remaining sub-busy periods $B_{>1}:=\bigcup_{i=2}^k B_i$, with probabilities equal to the corresponding time fractions,
\begin{align*}
P_{A_k}
=\frac{\E[A_k]}{\E[L^{k,m}]}
=\frac{k(1-\lambda)}{k+\lambda\E[I_m]},\quad
P_{B_1}
=\frac{\E[B_1]}{\E[L^{k,m}]}
=\frac{\lambda(\E[I_m]+1)}{k+\lambda\E[I_m]},\quad P_{B_{>1}}
=\frac{\sum_{i=2}^k \E[B_i]}{\E[L^{k,m}]}
=\frac{\lambda(k-1)}{k+\lambda\E[I_m]}.
\end{align*}

\begin{itemize}
\item\emph{Arrival during a standard sub-busy period.}
An arrival during $B_{>1}$ observes a standard $M/G/1$ FCFS queue conditioned on being busy. Hence
\begin{align*}
\E\!\left[T^{k,m}\mid B_{>1}\right]
&=\E[S]+\frac{\lambda\E[S^2]}{2\rho(1-\rho)} =1+\frac{\E[S^2]}{2(1-\lambda)} .
\end{align*}

Using the identity $\E[T^{M/G/1\text{-}L}]
=
(1-\rho)\E[S]
+
\rho\,\E[T^{M/G/1\text{-}L}\mid B],$
where $B$ denotes the event that an arrival finds the server busy in a standard $M/G/1$ FCFS queue, and noting that $\rho=\lambda$ and $\E[S]=1$, we obtain
\begin{align*}
\E\!\left[T^{k,m}\mid B_{>1}\right]
=
\frac{\E[T^{M/G/1\text{-}L}]-(1-\lambda)}{\lambda}.
\end{align*}

\item \emph{Arrival during the setup sub-busy period $B_1$.}
Let $T^{\mathrm{setup}}$ denote the stationary response time in an $M/G/1$ queue with setup $I_m$. The classical formula \cite{Welch1964,Harchol-Balter2013} gives
\begin{align*}
\E\!\left[T^{\mathrm{setup}}\right]
=\E[T^{M/G/1\text{-}L}]
+\frac{2\E[I_m]+\lambda\E[I_m^2]}{2\left(1+\lambda\E[I_m]\right)}.
\end{align*}
The busy fraction in this setup model is
\begin{align*}
\rho^{\mathrm{setup}}=\frac{\lambda\E[I_m]+\rho}{1+\lambda\E[I_m]}
=\frac{\lambda(\E[I_m]+1)}{1+\lambda\E[I_m]}.
\end{align*}
Conditioning on whether an arrival sees the setup system idle or busy yields $\E\!\left[T^{\mathrm{setup}}\right]
=\left(1-\rho^{\mathrm{setup}}\right)\E[S+I_m]+\rho^{\mathrm{setup}}\E\!\left[T^{k,m}\mid B_1\right].$
Therefore
\begin{align*}
\E\!\left[T^{k,m}\mid B_1\right]
=\frac{\E\!\left[T^{\mathrm{setup}}\right]-\left(1-\rho^{\mathrm{setup}}\right)\left(1+\E[I_m]\right)}{\rho^{\mathrm{setup}}}.
\end{align*}

\item \emph{Arrival during $A_k$.}
Condition on the number $j\in\{0,1,\ldots,k-1\}$ of jobs already accumulated upon arrival. Since exactly $k$ arrivals occur during $A_k$, this $j$ is uniform over $\{0,1,\ldots,k-1\}$. If $j=0$ (the job is $\text{Job}_1$), then
\begin{align*}
\E\!\left[T^{k,m}\mid A_k,\,0\right]
=\frac{k-1}{\lambda}+\E[I_m]+\E[S]
=\frac{k-1}{\lambda}+\E[I_m]+1.
\end{align*}
If $j\ge 1$ (the job is $\text{Job}_{j+1}$), then after waiting $\frac{k-j-1}{\lambda}$ more time for the threshold to be reached, it waits for $B_1$ and for $(j-1)$ standard sub-busy periods to finish, and then completes after its own service. Hence
\begin{align*}
\E\!\left[T^{k,m}\mid A_k,\,j\right]
=\frac{k-j-1}{\lambda}+\frac{\E[I_m]+1}{1-\lambda}+\frac{j-1}{1-\lambda}+1,
\quad\text{for } j=1,\ldots,k-1.
\end{align*}
Averaging over $j$ gives
\begin{align*}
\E\!\left[T^{k,m}\mid A_k\right]
&=\frac{1}{k}\sum_{j=0}^{k-1}\E\!\left[T^{k,m}\mid A_k,\,j\right] =1+\frac{k-1}{2\lambda}+\frac{k-1}{2(1-\lambda)}+\frac{1}{k}\E[I_m]+\frac{k-1}{k(1-\lambda)}\E[I_m].
\end{align*}
\end{itemize}

Finally, by total expectation, $\E[T^{k,m}]
=
P_{A_k}\,\E\!\left[T^{k,m}\mid A_k\right]
+
P_{B_1}\,\E\!\left[T^{k,m}\mid B_1\right]
+
P_{B_{>1}}\,\E\!\left[T^{k,m}\mid B_{>1}\right].$
Substituting the expressions above and simplifying yields
\begin{align*}
\E[T^{k,m}]
=
\E[T^{M/G/1\text{-}L}]
+
\frac{k(k-1)}{2\lambda\left(k+\lambda\E[I_m]\right)}
+
\frac{k\E[I_m]}{k+\lambda\E[I_m]}
+
\frac{\lambda\E[I_m^2]}{2\left(k+\lambda\E[I_m]\right)}.
\end{align*}

To obtain the throughput $X^{k,m}$, we divide the expected number of small-job completions per cycle, $\frac{km\mu_1}{\lambda}+m$, by the expected cycle length $\E[L^{k,m}]$.
\end{proof}

\lemTconcave*

\begin{proof}
Let $\mu:=\mu_1$, and let $H_m:=\sum_{j=1}^{m}1/j$ denote
the $m$th harmonic number, with $H_0:=0$.
For $m\ge 1$, write $x_m := \E[I_m] = \frac{1}{\mu}H_m$ and $y_m := \E[I_m^2].$
Define $D_m := k+\lambda x_m = k+\frac{\lambda}{\mu}H_m\quad\text{and}\quad
a := \frac{\lambda}{\mu} > 0.$

Start from
\begin{align}\label{eq:lem_a_1_1}
\E[T^{k,m}]
=
\E[T^{M/G/1\text{-}L}]
+
\frac{k(k-1)}{2\lambda D_m}
+
\frac{kx_m}{D_m}
+
\frac{\lambda y_m}{2D_m}.
\end{align}
Use $\frac{kx_m}{D_m}
=
\frac{k}{\lambda}\cdot\frac{\lambda x_m}{k+\lambda x_m}
=
\frac{k}{\lambda}\left(1-\frac{k}{D_m}\right)
=
\frac{k}{\lambda}-\frac{k^2}{\lambda D_m}$ therefore $\frac{k(k-1)}{2\lambda D_m}+\frac{kx_m}{D_m}
=
\frac{k}{\lambda}-\frac{k(k+1)}{2\lambda}\cdot\frac{1}{D_m},$
so \Cref{eq:lem_a_1_1} becomes
\begin{align}\label{eq:lem_a_1_2}
\E[T^{k,m}]
=
\E[T^{M/G/1\text{-}L}]
+
\frac{k}{\lambda}
-\frac{k(k+1)}{2\lambda}\cdot\frac{1}{D_m}
+\frac{\lambda}{2}\cdot\frac{y_m}{D_m}.
\end{align}
Thus it suffices to show
\begin{align}\label{eq:lem_a_1_3}
-\frac{1}{D_m} \text{ is concave increasing}\quad\text{and}\quad
 \frac{y_m}{D_m} \text{ is concave increasing}.
\end{align}

We first show $-1/D_m$ is concave increasing in $m$. Since $H_{m+1}-H_m=\frac{1}{m+1}$ decreases in $m$, $H_m$ is increasing and concave, hence $D_m=k+aH_m$ is increasing and concave. The function $g(x)=1/x$ is convex and decreasing on $(0,\infty)$, so $g(D_m)=1/D_m$ is convex \cite[Section~3.2.4]{Boyd2004}. Therefore $-1/D_m$ is concave. Since $D_m$ is increasing, $1/D_m$ is decreasing and $-1/D_m$ is increasing. Hence $-1/D_m$ is concave increasing.

Next, we show $y_m/D_m$ is also concave increasing in $m$. 
Use the identity
\begin{align*}
2\sum_{i=1}^{m}\frac{H_i}{i}
=
H_m^2+\sum_{i=1}^{m}\frac{1}{i^2}
=
H_m^2+H_m^{(2)}\quad\text{and}\quad
H_m^{(2)}:=\sum_{i=1}^{m}\frac{1}{i^2}.
\end{align*}
Hence $y_m
=
\E[I_m^2]
=
\frac{1}{\mu^2}\sum_{i=1}^{m}\frac{2}{i}H_i
=
\frac{1}{\mu^2}\left(H_m^2+H_m^{(2)}\right),$
therefore
\begin{align}\label{eq:lem_a_1_4}
\frac{y_m}{D_m}
=
\frac{1}{\mu^2}\cdot\frac{H_m^2+H_m^{(2)}}{k+aH_m}.
\end{align}
Thus, to show $\frac{y_m}{D_m}$ is concave increasing in $m$, it is enough to show that
\begin{align}\label{eq:lem_a_1_5}
r_m:=\frac{H_m^2+H_m^{(2)}}{k+aH_m}
\text{ is concave increasing in } m.
\end{align}
Let $h_m:=H_m,
p_m:=H_m^{(2)},
V_m:=k+ah_m$ and $r_m=\frac{h_m^2+p_m}{V_m}.$ We first show $r_m$ is increasing. Let $\delta_m:=h_{m+1}-h_m=\frac{1}{m+1}$. Then
\begin{align*}
h_{m+1}=h_m+\delta_m,
\quad
p_{m+1}=p_m+\delta_m^2
\quad \text{and}\quad
V_{m+1}=V_m+a\delta_m.
\end{align*}
A direct calculation gives $r_{m+1}-r_m
=
\frac{\delta_m}{V_mV_{m+1}}
\left(2kh_m+a(h_m^2-p_m)+2\delta_mV_m\right).$
Now $h_m^2-p_m\ge 0$ because $h_m^2-p_m
=
\left(\sum_{i=1}^{m}\frac{1}{i}\right)^2-\sum_{i=1}^{m}\frac{1}{i^2}
=
\sum_{i\ne j}\frac{1}{ij}\ge 0.$
All other factors are nonnegative, hence $r_{m+1}-r_m\ge 0$ and $r_m$ is increasing.

Next, we show $r_m$ is concave. Concavity means the increments decrease
\begin{align}\label{eq:lem_a_1_6}
(r_{m+1}-r_m)-(r_{m+2}-r_{m+1})\ge 0.
\end{align}
Let $t:=\delta_m=\frac{1}{m+1}$, and $\delta_{m+1}=\frac{1}{m+2}=\frac{t}{1+t}$. Introduce
\begin{align*}
u:=h_m-1=\sum_{i=2}^{m}\frac{1}{i}\ge 0\quad\text{and}\quad
s:=h_m^2-p_m=\sum_{i\ne j}\frac{1}{ij}\ge 0.
\end{align*}

A simple bound that hold for every $m\ge 1$ is $s\ge 2u$. Restrict the sum defining $s$ to the cross-terms involving index $1$,
\begin{align*}
s=\sum_{i\ne j}\frac{1}{ij}\ge \sum_{j=2}^{m}\left(\frac{1}{1\cdot j}+\frac{1}{j\cdot 1}\right) =2\sum_{j=2}^{m}\frac{1}{j}=2u.
\end{align*}

Using the recursions $h_{m+1}=h_m+t, p_{m+1}=p_m+t^2$, and the fact that $\delta_{m+1}=t/(1+t)$, one can expand the difference of increments and simplify to
\begin{align*}
(r_{m+1}-r_m)-(r_{m+2}-r_{m+1})
=
\frac{t^2}{(1+t)^2}\cdot\frac{-Q(u,s,t)}{V_mV_{m+1}V_{m+2}},
\end{align*}
where $V_{m}V_{m+1}V_{m+2}>0$ and $-Q(u,s,t)=(s-2u)L(u,t)+2F(u,t),
$
with $L(u,t)
=a^2t^2+a^2tu+4a^2t+a^2u+3a^2+akt+ak>0,$
and
\begin{align*}
F(u,t)
&=a^2\left[t^3(u+1)+t^2(u^2+6u+4)+t(2u^2+7u+2)+(u-1)\right] \\
&\quad+ak\left[t^3+t^2(3u+6)+t(u^2+8u+7)+(u^2+3u+1)\right]+k^2\left[t^2+tu+2t+u\right].
\end{align*}

Since $s-2u\ge 0$, we have $-Q(u,s,t)\ge 2F(u,t)$. To prove $F(u,t)>0$, it suffices to show that the first bracket $B(u,t)
:=t^3(u+1)+t^2(u^2+6u+4)+t(2u^2+7u+2)+(u-1)$
is strictly positive. Indeed, the remaining brackets in the definition of $F(u,t)$ are sums of nonnegative terms and are therefore strictly positive for $u\ge 0$ and $t>0$.

For $m\ge 4$, we have $u=H_m-1=\sum_{i=2}^{m}\frac{1}{i}\ge \frac{1}{2}+\frac{1}{3}+\frac{1}{4}=\frac{13}{12}>1,$
so $u-1>0$, and since all other terms in $B(u,t)$ are nonnegative, it follows that $B(u,t)>0$.

It remains to check $m=1,2,3$, a direct substitution yields $B\left(0,\frac12\right)=\frac{9}{8}, 
B\left(\frac12,\frac13\right)=\frac{85}{36}>0$ and $
B\left(\frac56,\frac14\right)=\frac{355}{128}.$
Therefore $B(u,t)>0$ for all $m\ge 1$, and hence $F(u,t)>0$, so $-Q(u,s,t)>0$, which implies \Cref{eq:lem_a_1_6}. Thus $r_m$ is concave. Together with monotonicity, $r_m$ is concave increasing and so is $y_m/D_m$.

Finally, from \Cref{eq:lem_a_1_2}, $-1/D_m$ is concave increasing and $y_m/D_m$ is concave increasing, and the coefficients are positive while the remaining terms are constant in $m$. Hence $\E[T^{k,m}]$ is increasing and concave in $m$.
\end{proof}

\lemLCconcave*

\begin{proof}
    We first consider the mean cycle length $\E[L^{k,m}]
=
\frac{k}{\lambda}
+
\frac{\E[I_m]+k}{1-\lambda}.$

Since $k$ and $\lambda$ are fixed, $\E[L^{k,m}]$ is an affine function of $\E[I_m]$. 
Because $\E[I_m]$ is increasing and concave in $m$, it immediately follows that $\E[L^{k,m}]$ is also increasing and concave in $m$.

Next, observe that $\lambda \E[L^{k,m}]
=
\lambda\left(
\frac{k}{\lambda}
+
\frac{\E[I_m]+k}{1-\lambda}
\right)
=
\frac{k+\lambda\E[I_m]}{1-\lambda}.$
Define $D_m := k+\lambda\E[I_m]$. Then $\lambda\E[L^{k,m}] = \frac{D_m}{1-\lambda}.$ Using the expression for $\E[T^{k,m}]$, we obtain
\begin{align*}
\E[C^{k,m}]
&=
\lambda\E[L^{k,m}]\cdot \E[T^{k,m}] =
\frac{D_m}{1-\lambda}
\left(
\E[T^{M/G/1\text{-}L}]
+
\frac{k(k-1)}{2\lambda D_m}
+
\frac{k\E[I_m]}{D_m}
+
\frac{\lambda\E[I_m^2]}{2D_m}
\right).
\end{align*}
Distributing $D_m/(1-\lambda)$ gives
\begin{align*}
\E[C^{k,m}]
=
\frac{\E[T^{M/G/1\text{-}L}]\,D_m}{1-\lambda}
+
\frac{k(k-1)}{2\lambda(1-\lambda)}
+
\frac{k\E[I_m]}{1-\lambda}
+
\frac{\lambda\E[I_m^2]}{2(1-\lambda)}.
\end{align*}
Substituting $D_m = k+\lambda\E[I_m]$ yields
\begin{align*}
\E[C^{k,m}]
&=
\underbrace{
\frac{k\E[T^{M/G/1\text{-}L}]}{1-\lambda}
+
\frac{k(k-1)}{2\lambda(1-\lambda)}
}_{\text{independent of } m}
+
\underbrace{
\frac{k+\lambda\E[T^{M/G/1\text{-}L}]}{1-\lambda}
}_{>0}
\,\E[I_m]
+
\underbrace{
\frac{\lambda}{2(1-\lambda)}
}_{>0}
\,\E[I_m^2].
\end{align*}

Thus $\E[C^{k,m}]$ is a positive affine combination of $\E[I_m]$ and $\E[I_m^2]$, plus a constant independent of $m$. Since both $\E[I_m]$ and $\E[I_m^2]$ are increasing and concave in $m$, and all coefficients are positive, it follows that $\E[C^{k,m}]$ is also increasing and concave in $m$.
\end{proof}

\lemConHypo*

\begin{proof}
It suffices to show that the consecutive secant slopes
\begin{align*}
s_k:=\frac{X^{k+1,n}-X^{k,n}}{\E[T^{k+1,n}]-\E[T^{k,n}]}
\end{align*}
are positive and decreasing in $k$. We give the argument for the general secant slopes between neighboring points $\left(\E[T^{k,n}],X^{k,n}\right)$, the comparison between the initial secant joining $\left(\E[T^{1,0}],0\right)$ and $\left(\E[T^{1,n}],X^{1,n}\right)$ and the next secant is similar.

From the explicit formulas for $\E[T^{k,n}]$ and $X^{k,n}$, we obtain
\begin{align*}
s_k=
\frac{2n\lambda^2(1-\lambda)\left(\mu_1\E[I_n]-1\right)}
{k+k^2-\lambda^2\E[I_n^2]+2k\lambda\E[I_n]+2\lambda^2\E[I_n]^2}.
\end{align*}
The numerator is positive because $\mu_1\E[I_n]=H_n>1$.

To show that the denominator is increasing and positive, define $g(k):=
k+k^2-\lambda^2\E[I_n^2]+2k\lambda\E[I_n]+2\lambda^2\E[I_n]^2.$
Then $g(k+1)-g(k)=2k+2+2\lambda\E[I_n]>0,$ so $g(k)$ is increasing in $k$. It therefore remains to show that $g(k)$ is positive. Since $g$ is increasing, it is enough to check $k=1$, $g(1)=2+2\lambda\E[I_n]+\lambda^2\left(2\E[I_n]^2-\E[I_n^2]\right).$
Using $\E[I_n]=H_n/\mu_1$ and $\E[I_n^2]=\left(H_n^2+H_n^{(2)}\right)/\mu_1^2$, where $H_n^{(2)}:=\sum_{i=1}^n i^{-2}$, we have
\begin{align*}
2\E[I_n]^2-\E[I_n^2]
&=
\frac{1}{\mu_1^2}\left(2H_n^2-\left(H_n^2+H_n^{(2)}\right)\right) =
\frac{1}{\mu_1^2}\left(H_n^2-H_n^{(2)}\right) =
\frac{1}{\mu_1^2}\sum_{i\ne j}\frac{1}{ij}
\ge 0.
\end{align*}
Hence $g(1)>0$, and therefore $g(k)>0$ for all $k\ge 1$.

Thus $s_k>0$ for all $k$, and since the numerator is constant in $k$ while the denominator $g(k)$ is increasing, $s_k$ is decreasing in $k$. Therefore the piecewise-linear interpolation of the points $\left(\E[T^{k,n}],X^{k,n}\right)$ is increasing and concave, equivalently its hypograph is convex.
\end{proof}

\section{Inequalities}
\label{app:2}





Let $f:\{0,1,\ldots,n\}\to\mathbb{R}$ be increasing and discretely concave, i.e.,
\begin{align*}
\Delta_k := f(k)-f(k-1)\ge 0\quad\text{and}\quad
\Delta_1\ge \Delta_2\ge \cdots \ge \Delta_n.
\end{align*}
In the next few lemmas, we study inequalities involving probabilistic or convex mixtures of such a function $f$. Throughout, we repeatedly use the telescoping representation $f(k)=f(0)+\sum_{i=1}^{k}\Delta_i,$ together with the monotonicity of the increments, $\Delta_i \ge \Delta_{i+1}$.

For each level $r\ge 1$, define
\begin{align*}
p_r := \frac{r\mu_1}{r\mu_1+\lambda}\quad\text{and}\quad
q_r := 1-p_r = \frac{\lambda}{r\mu_1+\lambda}.
\end{align*}
Define the truncated products (starting from level $m$):
\begin{align*}
\beta_r^{(m)} := \prod_{j=r+1}^{m} p_j\quad\text{for }
r=0,1,\ldots,m\quad\text{with}\quad
\beta_m^{(m)}:=1.
\end{align*}

From the above definition, $\beta_{r-1}^{(m)}
= \prod_{j=r}^{m} p_j
= p_r \prod_{j=r+1}^{m} p_j
= p_r \beta_r^{(m)}.$ Hence
\begin{align}\label{eq:lem3_0_3}
\beta_r^{(m)} q_r
= \beta_r^{(m)}(1-p_r)
= \beta_r^{(m)} - \beta_{r-1}^{(m)}.
\end{align}

A direct consequence is $\sum_{r=1}^{m} \beta_r^{(m)} q_r
= \sum_{r=1}^{m}\left(\beta_r^{(m)}-\beta_{r-1}^{(m)}\right)
= \beta_m^{(m)}-\beta_0^{(m)}
= 1-\beta_0^{(m)},$
hence
\begin{align}\label{eq:lem3_0_4}
\sum_{r=1}^{m} \beta_r^{(m)} q_r + \beta_0^{(m)} = 1.
\end{align}

\begin{lemma}
\label{lem:b_1}
Let
\begin{align*}
c := \frac{\lambda}{\mu_1} > 0\quad\text{and}\quad
p_r = \frac{r\mu_1}{r\mu_1+\lambda} = \frac{r}{r+c}.
\end{align*}
For every integer $m\ge 0$,
\begin{align}\label{eq:lem3_2_1}
\sum_{r=0}^{m-1}\beta_r^{(m)} = \frac{m}{1+c} = \frac{m\mu_1}{\mu_1+\lambda}.
\end{align}
\end{lemma}
\begin{proof}
We prove this by induction. Suppose \Cref{eq:lem3_2_1} is true for $m$, we prove that this also holds for the case $m+1$. For $r\le m$,
\begin{align*}
\beta_r^{(m+1)}
= \prod_{j=r+1}^{m+1}p_j
= \left[\prod_{j=r+1}^{m}p_j\right]p_{m+1}
= \beta_r^{(m)}p_{m+1}.
\end{align*}
Hence
\begin{align*}
\sum_{r=0}^{m}\beta_r^{(m+1)}
= p_{m+1}\sum_{r=0}^{m}\beta_r^{(m)}
= p_{m+1}\left[\sum_{r=0}^{m-1}\beta_r^{(m)}+\beta_m^{(m)}\right].
\end{align*}
By the induction hypothesis and $\beta_m^{(m)}=1$,
\begin{align*}
\sum_{r=0}^{m}\beta_r^{(m+1)}
= p_{m+1}\left(\frac{m}{1+c}+1\right)
= \frac{m+1}{m+1+c}\cdot \frac{m+1+c}{1+c}
= \frac{m+1}{1+c}.
\end{align*}
This is exactly \Cref{eq:lem3_2_1} for $m+1$.
\end{proof}

\begin{lemma}
\label{lem:b_2}
Fix integers $n>m\ge 1$. Let $f:\{0,1,\ldots,n\}\to\mathbb{R}$ be increasing and discretely concave. Let the mixing coefficient be
\begin{align*}
\alpha := \frac{m\mu_1}{n\mu_1+(n-m)\lambda}\in(0,1).
\end{align*}
Define the two quantities
\begin{align*}
T_1 &:= f(m),\\
T_2 &:= \alpha f(n) + (1-\alpha)\left[\sum_{r=1}^{m}\beta_r^{(m)}q_r f(r) + \beta_0^{(m)} f(0)\right].
\end{align*}
Then $T_1\ge T_2$.
\end{lemma}

\begin{proof}
Expand $T_1$ in terms of $\Delta_k$,
\begin{align}\label{eq:lem3_3_1}
T_1 = f(0) + \sum_{k=1}^{m}\Delta_k.
\end{align}
Next, we expand $T_2$ in terms of $\Delta_k$. Let $B^{(m)} := \sum_{r=1}^{m}\beta_r^{(m)}q_r f(r) + \beta_0^{(m)} f(0).$
Then $T_2 = \alpha f(n) + (1-\alpha)B^{(m)}.$ Expand $B^{(m)}$,
\begin{align*}
B^{(m)}
= \sum_{r=1}^{m}\beta_r^{(m)}q_r\left[f(0)+\sum_{k=1}^{r}\Delta_k\right] + \beta_0^{(m)} f(0).
\end{align*}
Group the $f(0)$ terms,
\begin{align*}
B^{(m)}
= f(0)\left[\sum_{r=1}^{m}\beta_r^{(m)}q_r + \beta_0^{(m)}\right]
+ \sum_{r=1}^{m}\beta_r^{(m)}q_r \sum_{k=1}^{r}\Delta_k.
\end{align*}
By \Cref{eq:lem3_0_4}, the bracket equals $1$. Hence $B^{(m)} = f(0) + \sum_{r=1}^{m}\beta_r^{(m)}q_r \sum_{k=1}^{r}\Delta_k.$

Next we switch the order of summation in the nested term $S := \sum_{r=1}^{m}\beta_r^{(m)}q_r \sum_{k=1}^{r}\Delta_k.$
A fixed $\Delta_k$ appears whenever $r\ge k$. Hence $S = \sum_{k=1}^{m}\Delta_k \sum_{r=k}^{m}\beta_r^{(m)}q_r.$
Now use \Cref{eq:lem3_0_3},
\begin{align*}
\sum_{r=k}^{m}\beta_r^{(m)}q_r
= \sum_{r=k}^{m}\left(\beta_r^{(m)}-\beta_{r-1}^{(m)}\right)
= \beta_m^{(m)}-\beta_{k-1}^{(m)}
= 1-\beta_{k-1}^{(m)}.
\end{align*}
Thus $S = \sum_{k=1}^{m}\Delta_k\left(1-\beta_{k-1}^{(m)}\right).$ Therefore $B^{(m)} = f(0) + \sum_{k=1}^{m}\Delta_k\left(1-\beta_{k-1}^{(m)}\right).$

Expand $f(n)$ using $f(n)=f(0)+\sum_{k=1}^{n}\Delta_k$ and plug into
$T_2$, we get
\begin{align*}
T_2
= \alpha\left[f(0)+\sum_{k=1}^{n}\Delta_k\right]
+ (1-\alpha)\left[f(0)+\sum_{k=1}^{m}\Delta_k(1-\beta_{k-1}^{(m)})\right].
\end{align*}
Group $f(0)$ again and notice that the coefficients sum to 1 $\alpha+(1-\alpha)=1$, so $f(0)$ remains.
Split $\sum_{k=1}^{n}\Delta_k=\sum_{k=1}^{m}\Delta_k+\sum_{k=m+1}^{n}\Delta_k$. Then the coefficient of $\Delta_k$ is, for $k\leq m$, $\alpha+(1-\alpha)(1-\beta_{k-1}^{(m)}) = 1-(1-\alpha)\beta_{k-1}^{(m)},$
and $\alpha$ for $k>m$. Therefore,
\begin{align}\label{eq:lem3_3_2}
T_2
= f(0)+\sum_{k=1}^{m}\Delta_k\left[1-(1-\alpha)\beta_{k-1}^{(m)}\right]
+ \alpha\sum_{k=m+1}^{n}\Delta_k.
\end{align}

Subtract \Cref{eq:lem3_3_2} from \Cref{eq:lem3_3_1}. The $f(0)$ cancels,
\begin{align*}
T_1-T_2
= \sum_{k=1}^{m}\Delta_k
-\sum_{k=1}^{m}\Delta_k\left[1-(1-\alpha)\beta_{k-1}^{(m)}\right]
-\alpha\sum_{k=m+1}^{n}\Delta_k.
\end{align*}
Simplify the first two sums,
\begin{align}\label{eq:lem3_3_3}
T_1-T_2
= (1-\alpha)\sum_{k=1}^{m}\beta_{k-1}^{(m)}\Delta_k
-\alpha\sum_{k=m+1}^{n}\Delta_k.
\end{align}

Now by concavity, we have $\Delta_k\ge \Delta_m$ for all $k\leq m$ and $\Delta_k\le \Delta_m$ for all $k\geq m+1$. Apply these to \Cref{eq:lem3_3_3},
\begin{align}\label{eq:lem3_3_4}
T_1-T_2
\ge \Delta_m\left[(1-\alpha)\sum_{k=1}^{m}\beta_{k-1}^{(m)} - \alpha(n-m)\right].
\end{align}
Thus it remains to evaluate the bracket.

Now apply \Cref{eq:lem3_2_1} of \Cref{lem:b_1} with $\sum_{k=1}^{m}\beta_{k-1}^{(m)}=\sum_{r=0}^{m-1}\beta_r^{(m)}$:
\begin{align}\label{eq:62}
\sum_{k=1}^{m}\beta_{k-1}^{(m)} = \frac{m\mu_1}{\mu_1+\lambda}.
\end{align}
Plug into the bracket in \Cref{eq:lem3_3_4},
\begin{align*}
(1-\alpha)\sum_{k=1}^{m}\beta_{k-1}^{(m)}-\alpha(n-m)
= (1-\alpha)\frac{m\mu_1}{\mu_1+\lambda}-\alpha(n-m).
\end{align*}
Now plug $\alpha=\frac{m\mu_1}{n\mu_1+(n-m)\lambda}$. First compute
\begin{align*}
1-\alpha
= \frac{n\mu_1+(n-m)\lambda-m\mu_1}{n\mu_1+(n-m)\lambda}
= \frac{(n-m)(\mu_1+\lambda)}{n\mu_1+(n-m)\lambda}.
\end{align*}
Therefore
\begin{align*}
(1-\alpha)\frac{m\mu_1}{\mu_1+\lambda}
= \frac{(n-m)(\mu_1+\lambda)}{n\mu_1+(n-m)\lambda}\cdot \frac{m\mu_1}{\mu_1+\lambda}
= \frac{(n-m)m\mu_1}{n\mu_1+(n-m)\lambda}
= \alpha(n-m).
\end{align*}
So the bracket equals $0$. Plugging into \Cref{eq:lem3_3_4} gives $T_1-T_2 \ge \Delta_m\cdot 0 = 0.$ Thus $T_1\ge T_2$.
\end{proof}

\begin{lemma}
\label{lem:b_3}
Let
\begin{align*}
c := \frac{\lambda}{\mu_1} > 0\quad\text{and}\quad
p_r = \frac{r\mu_1}{r\mu_1+\lambda} = \frac{r}{r+c}.
\end{align*}
For every integer $m\ge 0$,
\begin{align}\label{eq:lem3_4_1}
\sum_{r=0}^{m-1}\beta_r = \frac{m}{1+c}\beta_m.
\end{align}
\end{lemma}
\begin{proof}
We prove this by induction. Suppose \Cref{eq:lem3_4_1} is true for $m$, we prove that this also holds for the case $m+1$. We first write $\sum_{r=0}^{m}\beta_r
= \left(\sum_{r=0}^{m-1}\beta_r\right) + \beta_m.$
By the induction hypothesis,
\begin{align*}
\sum_{r=0}^{m}\beta_r
= \frac{m}{1+c}\beta_m + \beta_m
= \beta_m\left(\frac{m}{1+c} + 1\right)
= \beta_m \frac{m+1+c}{1+c}.
\end{align*}
Now observe that $\beta_m
= p_{m+1}\beta_{m+1}
= \frac{m+1}{m+1+c}\beta_{m+1}.$
Hence $\beta_m (m+1+c) = (m+1)\beta_{m+1}.$ Substituting into the previous expression yields $\sum_{r=0}^{m}\beta_r
= \frac{(m+1)\beta_{m+1}}{1+c}.$ This is exactly \Cref{eq:lem3_4_1} for $m+1$.
\end{proof}

\begin{lemma}
\label{lem:b_4}
    Let $f:\{0,1,\ldots,n\}\to\mathbb{R}$ be increasing and discretely concave. Let the mixing coefficient be
    \begin{align*}
        \alpha:=\frac{\ell}{n}\beta_{\ell}=\frac{\ell}{n}\prod_{r=\ell+1}^n\frac{r\mu_1}{r\mu_1+\lambda}
    \end{align*}

    Define the two quantities
    \begin{align*}
        T_1&=\sum_{r=\ell+1}^n\beta_rq_rf(r)+\beta_\ell f(\ell)\\
        T_2&=\alpha f(n)+(1-\alpha)\left[\sum_{r=1}^n\beta_rq_rf(r)+\beta_0f(0) \right].
    \end{align*}
    Then $T_1\geq T_2.$
\end{lemma}

\begin{proof}
Expand $T_1$ in terms of $\Delta_k$.
\begin{align*}
T_1 
&= \sum_{r=\ell+1}^{n} \beta_r q_r 
\left[
f(0) + \sum_{k=1}^{r} \Delta_k
\right]
+ \beta_\ell 
\left[
f(0) + \sum_{k=1}^{\ell} \Delta_k
\right].
\end{align*}

Group $f(0)$ terms,
\begin{align*}
T_1 
&= f(0)
\left[
\sum_{r=\ell+1}^{n} \beta_r q_r + \beta_\ell
\right]
+ \sum_{r=\ell+1}^{n} \beta_r q_r \sum_{k=1}^{r} \Delta_k
+ \beta_\ell \sum_{k=1}^{\ell} \Delta_k.
\end{align*}

By \Cref{eq:lem3_0_1}, $\beta_r q_r = \beta_r - \beta_{r-1}$,
\begin{align*}
\sum_{r=\ell+1}^{n} \beta_r q_r
= \sum_{r=\ell+1}^{n} (\beta_r - \beta_{r-1})
= \beta_n - \beta_\ell
= 1 - \beta_\ell.
\end{align*}
Hence $\sum_{r=\ell+1}^{n} \beta_r q_r + \beta_\ell = 1.$ Switch the order in the nested sum and define $S_1:=\sum_{r=\ell+1}^{n} \beta_r q_r \sum_{k=1}^{r} \Delta_k.$

A fixed $\Delta_k$ appears in the inner sum precisely when $k\le r$. Since $r$ runs from $\ell+1$ to $n$, we split into two cases:
\begin{itemize}
    \item If $k\le \ell$, then $k\le r$ for every $r=\ell+1,\ldots,n$, so $\Delta_k$ appears in every term of the outer sum.
    \item If $k\ge \ell+1$, then $k\le r$ forces $r=k,k+1,\ldots,n$.
\end{itemize}
Therefore,
\begin{align*}
S_1
&=
\sum_{k=1}^{\ell}\Delta_k\sum_{r=\ell+1}^{n}\beta_r q_r
+
\sum_{k=\ell+1}^{n}\Delta_k\sum_{r=k}^{n}\beta_r q_r.
\end{align*}

Again telescoping, $\sum_{r=k}^{n} \beta_r q_r
=
\sum_{r=k}^{n} (\beta_r - \beta_{r-1})
=
\beta_n - \beta_{k-1}
=
1 - \beta_{k-1}.$
Also, from the computation above, $\sum_{r=\ell+1}^{n}\beta_r q_r = 1-\beta_\ell.$ Substituting these into the expression for $S_1$ yields $S_1
=
(1-\beta_\ell)\sum_{k=1}^{\ell}\Delta_k
+
\sum_{k=\ell+1}^{n}\Delta_k(1-\beta_{k-1}).$

Thus,
\begin{align*}
T_1
&=
f(0)\cdot 1
+
S_1
+
\beta_\ell \sum_{k=1}^{\ell}\Delta_k =
f(0)
+
\left[(1-\beta_\ell)\sum_{k=1}^{\ell}\Delta_k
+
\sum_{k=\ell+1}^{n}\Delta_k(1-\beta_{k-1})\right]
+
\beta_\ell\sum_{k=1}^{\ell}\Delta_k.
\end{align*}
Combine the $\sum_{k=1}^{\ell}\Delta_k$ terms $(1-\beta_\ell)\sum_{k=1}^{\ell}\Delta_k+\beta_\ell\sum_{k=1}^{\ell}\Delta_k
=\sum_{k=1}^{\ell}\Delta_k.$ Hence
\begin{align}
\label{eq:lem3_5_3}
T_1
=
f(0)
+
\sum_{k=1}^{\ell} \Delta_k
+
\sum_{k=\ell+1}^{n} \Delta_k (1-\beta_{k-1}).
\end{align}

Next, we expand $T_2=\alpha f(n)+(1-\alpha B)$ in terms of $\Delta_k$, where 
\begin{align*}
B &:= \sum_{r=1}^{n} \beta_r q_r f(r) + \beta_0 f(0)=
\sum_{r=1}^{n} \beta_r q_r 
\left[
f(0) + \sum_{k=1}^{r} \Delta_k
\right]
+ \beta_0 f(0)=
f(0)\left[\sum_{r=1}^{n}\beta_r q_r+\beta_0\right]
+
\sum_{r=1}^{n}\beta_r q_r\sum_{k=1}^{r}\Delta_k.
\end{align*}

Since $\sum_{r=1}^{n} \beta_r q_r
= \sum_{r=1}^{n}(\beta_r-\beta_{r-1})
= \beta_n-\beta_0
= 1-\beta_0,$
we obtain $\sum_{r=1}^{n}\beta_r q_r+\beta_0 = 1,$
and therefore $B
=
f(0)
+
\sum_{r=1}^{n}\beta_r q_r\sum_{k=1}^{r}\Delta_k.$

Switch the order in the nested sum.
Define $S_2:=\sum_{r=1}^{n}\beta_r q_r\sum_{k=1}^{r}\Delta_k.$
A fixed $\Delta_k$ appears whenever $k\le r$, and now $r$ runs from $1$ to $n$, so for each fixed $k$ we have $r=k,k+1,\ldots,n$. Hence $S_2
=
\sum_{k=1}^{n}\Delta_k\sum_{r=k}^{n}\beta_r q_r.$
Using the same telescoping identity as above,
\begin{align*}
\sum_{r=k}^{n}\beta_r q_r
=
\sum_{r=k}^{n}(\beta_r-\beta_{r-1})
=
\beta_n-\beta_{k-1}
=
1-\beta_{k-1},
\end{align*}
Therefore $S_2
=
\sum_{k=1}^{n}\Delta_k(1-\beta_{k-1})$ and $B
=
f(0)
+
\sum_{k=1}^{n}\Delta_k(1-\beta_{k-1}).$

Thus
\begin{align}
\label{eq:lem3_5_4}
T_2
&=
\alpha 
\left[
f(0)+\sum_{k=1}^{n}\Delta_k
\right]
+
(1-\alpha)
\left[
f(0)+\sum_{k=1}^{n}\Delta_k(1-\beta_{k-1})
\right]=f(0)
+
\sum_{k=1}^{n}
\Delta_k
\left(
1 - (1-\alpha)\beta_{k-1}
\right).
\end{align}

Subtract \Cref{eq:lem3_5_4} from \Cref{eq:lem3_5_3}
\begin{align}
T_1 - T_2
&=
(1-\alpha)
\sum_{k=1}^{\ell}
\beta_{k-1}\Delta_k
-
\alpha
\sum_{k=\ell+1}^{n}
\beta_{k-1}\Delta_k.
\label{eq:lem3_5_1}
\end{align}

Since $f$ is concave, we have $\Delta_k\geq \Delta_\ell$ for $k\leq \ell$ and $\Delta_k\leq \Delta_\ell$ for $k\geq \ell+1$. Applying these bounds to \Cref{eq:lem3_5_1},
\begin{align}
T_1 - T_2
\ge
\Delta_\ell
\left[
(1-\alpha)
\sum_{k=1}^{\ell} \beta_{k-1}
-
\alpha
\sum_{k=\ell+1}^{n} \beta_{k-1}
\right].
\label{eq:lem3_5_2}
\end{align}

Using \Cref{eq:lem3_4_1} of \Cref{lem:b_3} we obtain
\begin{align*}
\sum_{k=1}^{\ell} \beta_{k-1}
=
\frac{\ell}{1+c}\beta_\ell\quad\text{and}\quad
\sum_{k=\ell+1}^{n} \beta_{k-1}
=
\frac{n}{1+c}
-
\frac{\ell}{1+c}\beta_\ell.
\end{align*}

Plugging into \Cref{eq:lem3_5_2} gives $T_1 - T_2
\ge
\Delta_\ell
\frac{1}{1+c}
\left(
\ell \beta_\ell - \alpha n
\right).$
With $\alpha = \frac{\ell}{n}\beta_\ell$, the bracket equals $0$. Hence $T_1 - T_2 \ge 0$.
\end{proof}

\end{document}